\documentclass[letterpaper,USenglish,cleveref]{lipics-v2021-modified}

\pdfoutput=1 
\hideLIPIcs  

\title{An Efficient Data Structure and Algorithm for Long-Match Query in Run-Length Compressed BWT}

\titlerunning{Long LEM Query in BWT-Runs Space} 

\author{Ahsan Sanaullah}{Department of Computer Science, University of Central Florida, Orlando, FL, USA}{ahsan.sanaullah@ucf.edu}{}{}

\author{Degui Zhi}{McWilliams School of Biomedical Informatics, University of Texas Health Science Center at Houston, Houston, TX, USA}{Degui.Zhi@uth.tmc.edu}{0000-0001-7754-1890}{}

\author{Shaojie Zhang}{Department of Computer Science, University of Central Florida, Orlando, FL, USA}{shzhang@cs.ucf.edu}{0000-0002-4051-5549}{}

\authorrunning{A. Sanaullah, D. Zhi, and S. Zhang} 

\Copyright{Ahsan Sanaullah, Degui Zhi, and Shaojie Zhang} 

\ccsdesc[500]{Theory of computation~Pattern matching}
\ccsdesc[500]{Theory of computation~Data compression}

\keywords{BWT, LEM, Long LEM, MEM, Run Length Compressed BWT, Move Data Structure, Pangenome} 

\funding{This work was supported by the National Institutes of Health grant R01 HG010086.}

\nolinenumbers 

\usepackage{mathtools}
\usepackage{cite}

\EventEditors{John Q. Open and Joan R. Access}
\EventNoEds{2}
\EventLongTitle{42nd Conference on Very Important Topics (CVIT 2016)}
\EventShortTitle{CVIT 2016}
\EventAcronym{CVIT}
\EventYear{2016}
\EventDate{December 24--27, 2016}
\EventLocation{Little Whinging, United Kingdom}
\EventLogo{}
\SeriesVolume{42}
\ArticleNo{23}

\begin{document}

\maketitle

\begin{abstract}
String matching problems in bioinformatics are typically for finding exact substring matches between a query and a reference text. Previous formulations often focus on maximum exact matches (MEMs). However, multiple occurrences of substrings of the query in the text that are long enough but not maximal may not be captured by MEMs. Such long matches can be informative, especially when the text is a collection of similar sequences such as genomes. In this paper, we describe a new type of match between a pattern and a text that aren't necessarily maximal in the query, but still contain useful matching information: locally maximal exact matches (LEMs). There are usually a large amount of LEMs, so we only consider those above some length threshold $\mathcal{L}$. These are referred to as long LEMs. The purpose of long LEMs is to capture substring matches between a query and a text that are not necessarily maximal in the pattern but still long enough to be important. Therefore efficient long LEMs finding algorithms are desired for these datasets. However, these datasets are too large to query on traditional string indexes. Fortunately, these datasets are very repetitive. Recently, compressed string indexes that take advantage of the redundancy in the data but retain efficient querying capability have been proposed as a solution. We therefore give an efficient algorithm for computing all the long LEMs of a query and a text in a BWT runs compressed string index. We describe an $O(m+occ)$ expected time algorithm that relies on an $O(r)$ words space string index for outputting all long LEMs of a pattern with respect to a text given the matching statistics of the pattern with respect to the text. Here $m$ is the length of the query, $occ$ is the number of long LEMs outputted, and $r$ is the number of runs in the BWT of the text. The $O(r)$ space string index we describe relies on an adaptation of the move data structure by Nishimoto and Tabei. We are able to support $LCP[i]$ queries in constant time given $SA[i]$. In other words, we answer $PLCP[i]$ queries in constant time. Long LEMs may provide useful similarity information between a pattern and a text that MEMs may ignore. This information is particularly useful in pangenome and biobank scale haplotype panel contexts.
\end{abstract}

\clearpage

\section{Introduction}\label{introsect}
Bioinformatics sequence data is often large and very repetitive. Furthermore, efficient matching queries on the data are frequently needed for many biological analyses. Therefore, bioinformatics problems have incentivized and profited from the development of efficient string indexes. The Burrows-Wheeler transform (BWT) has thus been used in bioinformatics algorithms. The BWT is a permutation of a text that has found wide use in string indexing and data compression~\cite{burrows1994block}. Position $i$ in the BWT of the text is essentially the character behind the $i$-th lexicographically smallest suffix of the text. Therefore, adjacent characters of the BWT correspond to highly locally similar suffixes of the text. Therefore, the BWT of highly repetitive texts tends to have large runs of one character, with an overall small number of runs. The BWT of highly repetitive texts therefore compresses well. In fact, the number of runs in BWT, $r$, is sometimes used as a measure of the repetitiveness of a string~\cite{navarropart1_10.1145/3434399}. Finally, given only the BWT of a text, the text can be reconstructed in linear time~\cite{burrows1994block} and the BWT of a text can be constructed in linear time by construction of the suffix array~\cite{KO2005143}. The BWT ordering also allows efficient string indexes. In other words, given a pattern, find all occurrences of the pattern within the text. String indexes have been shown that output all occurrences of a pattern (a locate query) in space linear to the product of the length of the text and the size of the alphabet and time linear to the sum of the length of the pattern and the number of occurrences~\cite{FMIndex10.1145/1082036.1082039}.

Compressed string indexes have also been shown~\cite{FMIndex10.1145/1082036.1082039,10.1007/978-3-319-19929-0_3,doi:10.1089/cmb.2009.0169}. These indexes output all occurrences of a pattern in space sublinear to the size of the text. Although the time complexity of locating these occurrences is not linear in the length of the pattern and the number of occurrences, they are typically independent of the length of the text barring logarithmic factors and close to linear in the length of the pattern and number of occurrences. In particular for highly repetitive texts, the space of the index can be much smaller than the space of the text. Notably, recent compressed string indexes have achieved space linear to the number of runs in the BWT ($r$)~\cite{rindex10.1145/3375890,nishimoto_et_al:LIPIcs.ICALP.2021.101}. The r-index by Gagie et al. was the first compressed string index offering close to linear time locate queries in $O(r)$ space~\cite{rindex10.1145/3375890}. Nishimoto and Tabei recently improved on this result with their OptBWTR, which achieves linear time locate queries for texts with alphabets of size polylogarithmic in the length of the text. OptBWTR relies on the move data structure, which was also introduced by Nishimoto and Tabei~\cite{nishimoto_et_al:LIPIcs.ICALP.2021.101}.

Compressed string indexes have been fruitfully applied to the growing collection of bioinformatics data. Over the past two decades, large collections of genomics data have grown increasingly larger in size. For example, the UK Biobank has whole genome sequencing data of roughly one million haplotypes~\cite{UKBLi2023.12.06.23299426}, and the All of Us program has released whole genome sequencing data of half a million haplotypes~\cite{all2024genomic}. Furthermore, recent arguments have been made that a human reference pangenome should be used instead of a singular human reference genome to avoid reference bias in downstream analyses~\cite{miga2021need,taylor2024beyond,singh2022reference}. The Human Pangenome Reference Consortium has released a draft human pangenome reference of more than two hundred high quality phased diploid assemblies and is planning to release over three hundred and fifty in the final release~\cite{liao2023draft,draftpangenomerelease2}. The UK Biobank whole genome sequencing data has 1.5 billion variants, the All of Us whole genome sequencing data has 1 billion variants, and the typical diploid assembly in the draft human pangenome has 6 billion bases. Therefore, these datasets have 1,500 trillion, 250 trillion, and 1.2 trillion characters each respectively. However, while very large, these datasets are very repetitive. Furthermore, queries on these datasets are frequently needed for biological applications including read mapping\cite{bwa10.1093/bioinformatics/btp324,langmead2012fast}, read alignment~\cite{10.1093/bioinformatics/btp698}, read classification for metagenomes~\cite{song2024centrifuger,ahmed2023spumoni,Depuydt2025.02.25.640119} or pangenomes~\cite{nate10.1007/978-3-031-90252-9_12}. Many general purpose compressed string indexes have also been implemented for exact pattern matching and matching substrings of the pattern~\cite{monidoi:10.1089/cmb.2021.0290,zakeri2024movi,bmovedepuydt_et_al:LIPIcs.WABI.2024.10,ropebwt310.1093/bioinformatics/btae717}. These indexes may compute the maximal exact matches (MEMs) of the pattern with respect to the text. MEMs are matches between the pattern and the text that cannot be extended in the pattern. 

While MEMs typically refers to matches that are maximal in the pattern, matches that are simultaneously maximal in the pattern and the text may sometimes be desired. Notably, in two data structures related to the BWT, algorithms for outputting matches that are simultaneously maximal in the pattern and the text have already been developed. These data structures are the positional Burrows-Wheeler transform (PBWT) and the graph Burrows-Wheeler transform (GBWT)~\cite{PBWT,GBWT}. In the PBWT and GBWT, matches that cannot be extended in the pattern are referred to as set maximal matches, and matches that cannot be simultaneously extended in the pattern are referred to as locally maximal matches. Locally maximal matches that are longer than some length threshold $\mathcal{L}$ are referred to as $\mathcal{L}$-long matches, or long matches for short. Algorithms for outputting set maximal matches and long matches have been published in the PBWT in uncompressed~\cite{PBWT,naseri10.1093/bioinformatics/btz347,dpbwt} and compressed space~\cite{mupbwt10.1093/bioinformatics/btad552,dynmupbwt10.1007/978-3-031-90252-9_13,sylpbwt10.1093/bioinformatics/btac734,smempbwt10.1007/978-3-031-43980-3_8}. Algorithms for outputting these matches have also been published for the GBWT in compressed space~\cite{gbwtquerySanaullah2025.02.03.634410}. 

In this paper, we use these concepts in the traditional pattern and text context, and name matches that cannot be simultaneously extended in the pattern and the text locally maximal exact matches (LEMs). LEMs that are longer than $\mathcal{L}$ are long LEMs. The distinction between matches that do not extend in the pattern and matches that do not extend simultaneously in the pattern and the text has been made before. Notably, in ropebwt3 MEMs refers to LEMs of our paper and super maximal exact matches (SMEMs) refers to MEMs of our paper~\cite{ropebwt310.1093/bioinformatics/btae717}. The term SMEM has been used in place of MEM in a few papers to avoid the confusion in terminology~\cite{smempbwt10.1007/978-3-031-43980-3_8,Depuydt2025.02.25.640119,ropebwt310.1093/bioinformatics/btae717,mupbwt10.1093/bioinformatics/btad552}, however MEM is still the most common term by far for matches that cannot be extended in the pattern. 

In this work, we describe an algorithm for outputting all long LEMs of a pattern with respect to a text in $O(m+occ)$ expected time given the matching statistics of the pattern with respect to the text, where $m$ is the length of the pattern and $occ$ is the number of long LEMs it has with respect to the text. In order to do so, we modify the OptBWTR data structure of Nishimoto and Tabei to also compute $LCP[i]$ given $SA[i]$ (i.e. compute $PLCP$). We name this modified OptBWTR, \textit{OptBWTRL} (i.e. OptBWTR for long LEMs or OptBWTR with LCP). OptBWTRL maintains the $O(r)$ words space complexity of OptBWTR and computes $\phi[i]$ and $PLCP[i]$ in constant time. The long LEM finding algorithm also requires as input an OptBWTRL of the text. We also discuss possible future work related to this paper, including avenues for improving the results, utilization of constant time $PLCP$ computation to speed up matching statistics computation, and biological applications of long LEMs. Long LEMs may have many biological applications, from identity by descent segment detection and local ancestry inferences, to seeds or anchors for approximate matching algorithms for genome to genome alignment, genome to pangenome, read to genome or other alignments. In this paper, our main contributions are the following:

\begin{itemize}
    \item \textbf{OptBWTRL:} OptBWTRL is an $O(r)$ words space data structure that maintains the capabilities of OptBWTR and adds the ability to compute $\phi,PLCP,$ and long LEMs efficiently. $r$ is the number of runs in the BWT of the text.
    \begin{itemize}
        \item \textbf{PLCP:} OptBWTRL enables constant time $PLCP[i]$ computation in $O(r)$ space. Note that $PLCP[SA[j]]=LCP[j]$, therefore $PLCP$ computation in constant time allows $LCP[j]$ computation in constant time given $SA[j]$.
        \item \textbf{Long LEM Query:} We describe an $O(m+occ)$ expected time long LEM query for pattern $P$ and text $T$ given the matching statistics of $P$ with respect to $T$. The underlying index (OptBWTRL) uses $O(r)$ space. $m$ is the length of $P$ and $occ$ is the number of long LEMs $P$ has with respect to $T$. A deterministic time bound for a similar algorithm we show is $O\left(m+occ\sqrt{\frac{\log occ}{\log\log occ}}\right)$. 
    \end{itemize}
    \item \textbf{Long LEM Query with random access to the text:} Given $O(t_{RA})$ time random access to the text and a BWT related index, algorithms for computing matching statistics efficiently are known. Therefore, our long LEM query algorithm results in the following.
    \begin{itemize}
        \item \textbf{In Uncompressed Space:} An algorithm for long LEM query in $O(m+occ)$ expected time in uncompressed string indexes such as the FM Index (Corollary 3.2, variant with $O(n\sigma)$ space, where $n$ is the length of the text and $\sigma$ is the size of the alphabet)~\cite{FMIndex10.1145/1082036.1082039}.
        \item \textbf{In Compressed Space:} An algorithm for long LEM query in $O(m\log\frac{n}{\delta}+occ)$ expected time in $O(r+\delta\log\frac{n}{\delta})$ space given a block tree~\cite{delta10.1007/978-3-030-61792-9_17,blocktreeBELAZZOUGUI20211} (with random access to the text in $O(\log\frac{n}{\delta})$ time in $O(n\log\frac{n}{\delta})$ space) and an OptBWTRL of the text. 
    \end{itemize}
\end{itemize}

\section{Background}
In this section, we review definitions used throughout the rest of the paper. We begin with strings, then in~\Cref{bwtback}, we review BWT related concepts. In~\Cref{nishback}, we give a short overview of the results of Nishimoto and Tabei in~\cite{nishimoto_et_al:LIPIcs.ICALP.2021.101}. Matching statistics are reviewed in~\Cref{msback}. Finally we review maximal exact matches (MEMs) and define locally maximal exact matches (LEMs) in~\Cref{lemback}. 

Let $\Sigma = \{1,2,3,\dots,\sigma\}$ be an ordered alphabet of size $\sigma$. The size (number of characters it contains) of a string $T$ is represented by $|T|$. $T$ refers to a text of length $n$ ($|T| = n$) where the last character is $\$$. The character $\$$ is lexicographically smaller than all other characters in $T$ and occurs only in the last position of $T$. The $i$-th character of $T$ is $T[i]$, $i \in [1,n]$. $T[i,j]$ refers to the substring of $T$ that starts at position $i$ and ends at position $j$, inclusive ($T[i,j] = T[i]T[i+1]T[i+1]\dots T[j]$). Prefix $i$ of $T$ is the string $T[1,i]$, suffix $i$ of $T$ is $T[i,n]$. The longest common prefix of two strings $T$ and $T'$ is referred to by $lcp(T, T')$. $|lcp(T,T')|$ is the largest value $i$ s.t. $i\leq\min(|T|,|T'|)$ and $T[1,i] = T'[1,i]$ (then, $lcp(T,T') = T[1,i] = T'[1,i]$). A string $T'$ being lexicographically smaller than $T$ is represented by $T'\prec T$. If $T' = T$, $T'\nprec T$ and $T\nprec T'$. If $T' \neq T$, $T'\prec T$ iff $T' = lcp(T, T')$ or $T'[|lcp(T,T')|+1] < T[|lcp(T,T')|+1]$. 

\subsection{Burrows-Wheeler Transform}\label{bwtback}
The Suffix Array ($SA$) of a text $T$ is an array of length $n = |T|$ where the $i$-th position stores the index of the $i$-th lexicographically smallest suffix of $T$. Therefore, $T[SA[1],n] \prec T[SA[2],n] \prec T[SA[3],n] \prec \dots \prec T[SA[n],n]$. The Burrows-Wheeler Transform (BWT) of a text $T$ is a string of length $n$ where the $i$-th character in the string is the $SA[i]-1$-th character of $T$ (the $n$-th character if $SA[i] = 1$). The $LF$ array is an array of length $n$ that stores the position of the previous suffix in the suffix array, $LF[i] = j$ s.t. $SA[j] = SA[i] - 1$ for all $SA[j]\in[1,n-1]$, $LF[i]=j$ s.t. $SA[j] = n$ for $SA[i] = 1$. The $\phi$ array stores at position $i$, the suffix above suffix $i$ in the suffix array, i.e. if $SA[k] = i$, $\phi[i] = SA[k-1]$ ($\phi[i] = SA[n]$ if $i = SA[1]$). The $\phi^{-1}$ array stores at position $i$, the suffix below suffix $i$ in the suffix array, i.e. if $SA[k] = i$, $\phi^{-1}[i] = SA[k+1]$ ($\phi^{-1}[i] = SA[1]$ if $i = SA[n]$). Therefore, $\phi[\phi^{-1}[i]] = i$ and $\phi^{-1}[\phi[i]] = i$. The $LCP$ array is an array of length $n$ where $LCP[i]$ stores the length of the longest common prefix of suffix $SA[i]$ and $SA[i-1]$. $LCP[1] = 0$ and for $i\in[2,n]$, $LCP[i] = |lcp(T[SA[i],n], T[SA[i-1],n])|$. The $PLCP$ (permuted $LCP$) array is an array of length $n$ where the $LCP$ array is stored by suffix index. Therefore, if $SA[i] = j$, $PLCP[j] = LCP[i]=|lcp(T[j,n],T[\phi[j],n])|$. Finally, the inverse suffix array, $ISA$, is an array of length $n$ that stores at position $i$ the position of suffix $i$ in the suffix array, if $ISA[i] = j$, $SA[j] = i$. $ISA[SA[i]] = i$ and $SA[ISA[i]] = i$. The $SA, LF, \phi, \phi^{-1},$ and $ISA$ arrays are permutations of the integers in $[1,n]$. $SA$ and $ISA$ are inverses of each other and $\phi$ and $\phi^{-1}$ are inverses of each other.

The run-length burrows-wheeler transform (RLBWT) is the run-length encoding of the BWT of a text. Call $L$ the BWT of text $T$. Then, $L$ is partitioned into $r$ nonempty substrings $L_1,L_2,\dots,L_r$. $L_i$ is a substring of $L$ corresponding to the $i$-th run of $L$. A run is a maximal repetition of the same character in $L$. Therefore, $L_i[1] = L_i[2] = \dots = L_i[|L_i|]$ for all $i\in[1,r]$ and $L_i[1] \neq L_{i+1}[1]$ for all $i\in[1,r-1]$. $l_i$ is the starting position of the run $L_i$ in $L$. The RLBWT is represented as $r$ pairs $(L_i[1], l_i)$ for $i\in[1,r]$. All of these structures can be seen in~\cref{figBWT} for a text $T=missisismississippi\$$.

\begin{figure}
    \centering
    \includegraphics[width=\linewidth]{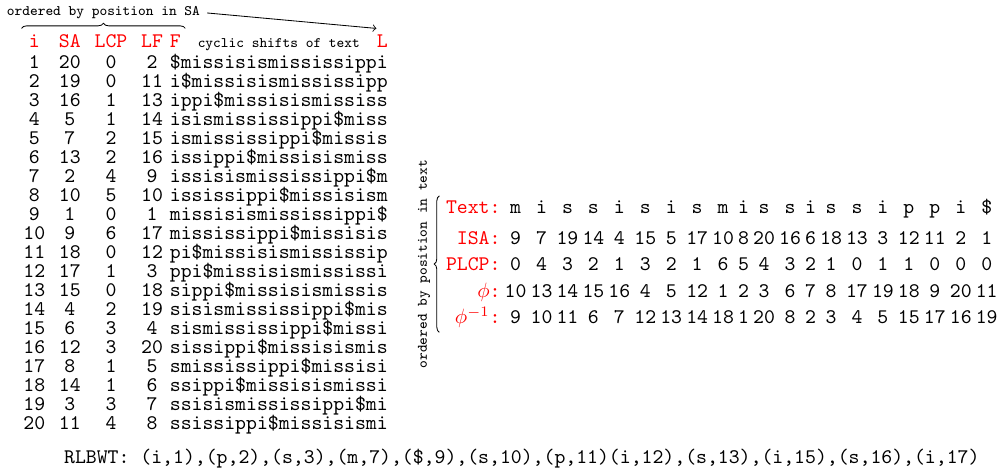}
    \caption{BWT and related structures for a text $T=missisismississippi\$$. $SA,LCP,LF,F,$ and $L$ are orderd by position in $SA$ while $ISA,PLCP,\phi,$ and $\phi^{-1}$ are ordered by position in the text. }
    \label{figBWT}
\end{figure}

\subsection{Move Data Structure}\label{nishback}
The move data structure is a data structure for representing a permutation of a contiguous range of integers efficiently. It was introduced by Nishimoto and Tabei~\cite{nishimoto_et_al:LIPIcs.ICALP.2021.101}. In the original introduction, the structure was described for a permutation of $[1,n]$. This of course may be extended to any bijective function from a contiguous range of integers to another contiguous range of integers. The move data structure takes space proportional to the number of intervals conserved in the function. An interval is conserved in a bijective function from a contiguous range of integers to another contiguous range of integers if for any $i,j$ in the interval, $f(i) - f(j) = i - j$ (therefore, $f(i) = f(j) + i - j$ and $f(i)-i = f(j)-j$). The move data structure computes the represented function in constant time. The important arrays, $LF$ and $\phi^{-1}$, are permutations of $[1,n]$ with $O(r)$ conserved intervals, where $r$ is the number of runs in the BWT. Therefore, Nishimoto and Tabei define the OptBWTR data structures using move data structures. OptBWTR  supports efficient count and locate queries in BWT-runs compressed space. Below, we more formally review some of the results from their paper~\cite{nishimoto_et_al:LIPIcs.ICALP.2021.101}.

\subsubsection{Disjoint Interval Sequence}
$I=(p_1,q_1),(p_2,q_2),\dots,(p_k,q_k)$ is a sequence of $k$ pairs of integers. Let $p_{k+1} = n+1$. Then $i$ is a \textit{disjoint interval sequence} iff there exists a permutation $\pi$ of $[1,k]$ s.t. (i) $p_1 = 1 < p_2 < \dots < p_k \leq n$, (ii) $q_{\pi[1]} = 1$, and (iii) $q_{\pi[i]} = q_{\pi[i-1]} + (p_{\pi[i-1]+1} - p_{\pi[i-1]})$. $[p_i, p_{i+1} - 1]$ is referred to as the $i$-th input interval, and $[q_{i}, q_{i} + (p_{i+1}-p_i) - 1]$ as the $i$-th output interval. The input intervals don't overlap, and their union is $[1,n]$. The output intervals don't overlap and their union is $[1,n]$.

A \textit{move query} on a disjoint interval sequence $I$ takes as input $(i, x)$, where $i$ is an index in $[1,n]$ and $x$ is the index of the input interval sequence that contains it, $i\in[1,n]$ and $p_x \leq i < p_{x+1}$ and $x\in[1,k]$. The move query outputs $(i', x')$ where $i' = q_x + (i - p_x)$ and $p_{x'} \leq i' < p_{x'+1}$, i.e. $i'$ is the mapping of position $i$ from the input to output intervals by $I$ and $x'$ is the index of the input interval that contains $i'$.

$f$, a permutation of $[1,n]$ with $k$ conserved intervals, can be represented by a disjoint interval sequence where the input intervals are the conserved intervals and the output intervals are the mapping of the input intervals by $f$. Then, a move query of $(i,x)$ returning $(i',x')$ computes $f$ by $f(i) = i'$.

Nishimoto and Tabei show that a move queries on a disjoint interval sequence of $k$ input intervals (and therefore $k$ output intervals) can be computed in constant time and $O(k)$ space with the move data structure. The move data structure is built by splitting the $k$ input intervals of $I$ into at most $2k$ intervals. This results in a disjoint interval sequence of at most $2k$ input intervals (and an equivalent number of output intervals) that represents the same permutation as the original disjoint interval sequence. The split interval sequence of $i$ that the move data structure is built on is referred to as a balanced interval sequence. The notation for a balanced interval sequence of $I$ is $B(I)$, and the notation for a move data structure of $I$ is $F(I)$. In this paper, we occasionally use input interval of $F(I)$ as shorthand for input interval of $B(I)$ (for example, $i$-th input interval of a move data structure refers to the $i$-th input interval of the balanced interval sequence it was built on). Brown et al. extend the balanced interval sequence result of Nishimoto and Tabei to splitting $I$'s $k$ intervals into at most $k+\frac{k}{d-1}$ intervals, resulting in a move data structure with $O(d)$ time move query computation for any $d\geq 2$~\cite{brown_et_al:LIPIcs.SEA.2022.16}.  

\subsubsection{OptBWTR}
The arrays $LF$ and $\phi^{-1}$ are permutations of $[1,n]$ with $O(r)$ conserved intervals. For $LF$, a conserved interval is within a run in the BWT. For $\phi^{-1}$, a conserved interval is a range of suffixes of $T$ that don't occur at the bottom of a run in the BWT (except the first position of the interval may be at the bottom of a run). Therefore, Nishimoto and Tabei define the OptBWTR data structure as the combination of the move data structures of the $LF$ and $\phi^{-1}$ functions along with a rank-select data structure on an $O(r)$ length string $L_{first}$. OptBWTR supports $O(m\log\log_w\sigma)$ time count queries and $O(m\log\log_w\sigma + occ)$ time locate queries in $O(r)$ words of space, where $r$ is the number of runs in the BWT of the text, $m$ is the length of the pattern, $occ$ is the number of occurrences of the pattern in the text, $w$ is the word size, $\sigma$ is the size of the alphabet (Theorem 9 of~\cite{nishimoto_et_al:LIPIcs.ICALP.2021.101}). The input intervals of $B(I_{LF})$, the disjoint interval sequence of the move data structure of $LF$, are contained within a run in the BWT. Call the $i$-th input interval of $B(I_{LF})$ $[p_i, p_{i+1}-1]$. Then, $L_{first} = L[p_1]L[p_2]L[p_3]\dots L[p_k]$, where $k \leq 2r$ is the number of input intervals of $B(I_{LF})$, $L$ is the BWT of $T$, and $\forall i\in[1,k], j\in[p_i, p_{i+1}-1] L[j] = L[p_i]$. Call $B(I_{SA})$ the disjoint interval sequence of the move data structure of $\phi^{-1}$, and $[p^-_i, p^-_{i+1}]$ its $i$-th input interval. OptBWTR is composed of 
\begin{itemize}
    \item move data structures for $LF$ and $\phi^{-1}$ ($F(I_{LF})$ and $F(I_{SA})$ respectively),
    \item a rank-select data structure on $L_{first}$ ($R(L_{first})$),
    \item samples of the $SA$ at the beginning of input intervals of the LF move data structure ($SA^+$, where $SA^+[i] = SA[p_i]$), and
    \item the index of the input interval of the $\phi^{-1}$ move data structure that contains each $SA$ sample in $SA^+$ ($SA^+_{index}$, where $SA^+_{index}= y \iff SA^+[i]\in [p^-_y, p^-_{y+1}]$).
\end{itemize}

\subsection{Matching Statistics}\label{msback}
The matching statistics of a pattern $P$ with respect to a text $T$ represents information on the local similarity of the pattern to the text. The matching statistics of $P$ with respect to $T$, $\prescript{}{P}{MS}_T$, is an array of length $|P| = m$ that stores at position $i$ three values: $\prescript{}{P}{MS}_T[i].len$, $\prescript{}{P}{MS}_T[i].suff$, and $\prescript{}{P}{MS}_T[i].row$. $\prescript{}{P}{MS}_T[i].len$ is the length of the longest substring of $P$ starting at $i$ that occurs in $T$. $\prescript{}{P}{MS}_T[i].suff$ is a suffix of $T$ that has a longest common prefix with $P[i,m]$ of length $\prescript{}{P}{MS}_T[i].len$ (or equivalently, $\prescript{}{P}{MS}_T[i].suff$ is the starting position of an occurrence of $P[i,i+\prescript{}{P}{MS}_T[i].len-1]$ in $T$). $\prescript{}{P}{MS}_T[i].row$ is the index in the SA of $T$ that has value $\prescript{}{P}{MS}_T.suff$. Formally for all $i\in[1,m]$, 
\begin{itemize}
    \item $\prescript{}{P}{MS}_T[i].len = \max_{j\in[1,|T|]} |lcp(P[i,m], T[j,|T|])|$,
    \item $|lcp(T[\prescript{}{P}{MS}_T[i].suff, n], P[i, m])| = \prescript{}{P}{MS}_T[i].len$, and
    \item $SA[\prescript{}{P}{MS}_T[i].row] = \prescript{}{P}{MS}_T[i].suff$.
\end{itemize}
When $P$ and $T$ are clear from the context, we omit them from $\prescript{}{P}{MS}_T$ and refer to the matching statistics of $P$ with respect to $T$ as $MS$.
\subsection{Maximal and Locally Maximal Exact Matches}\label{lemback}
For a pattern $P$ and a text $T$ ($|P| = m, |T| = n)$, a maximal exact match (MEM), $P[i,j] = T[i',j']$, is a match between $P$ and $T$ that cannot be extended left or right in the pattern. Formally, ($i = 1$ or $P[i-1,j]$ doesn't occur in $T$) and ($j = m$ or $P[i, j+1]$ doesn't occur in $T$). A MEM can be fully specified by the triple $(i,i',k)$ where $k=j-i+1$ is the length of the match and $i$ and $i'$ are the starting positions of the match in the pattern and the text respectively. 

For a pattern $P$ and a text $T$, a locally maximal exact match (LEM), $P[i,j] = T[i',j']$, is a match between $P$ and $T$ that cannot be simultaneously extended in the pattern and the text. The match cannot be simultaneously extended left in the pattern and the text. Likewise, it cannot be simultaneously extended right in the pattern and the text. Formally, ($i = 1$ or $i' = 1$ or $P[i-1,j] \neq T[i'-1,j']$) and ($j = m$ or $j' = n$ or $P[i,j+1]\neq T[i',j'+1]$). A LEM can also be fully specified by the triple $(i,i',k)$ where $k$ is the length of the LEM and $k=j-i+1$. For some length threshold $\mathcal{L}$, a long LEM is a LEM with length at least $\mathcal{L}$. See~\Cref{memslemsfig} for a depiction of MEMs and LEMs in a text representing a pangenome.

\begin{figure}
    \centering
    \includegraphics[width=0.91\linewidth]{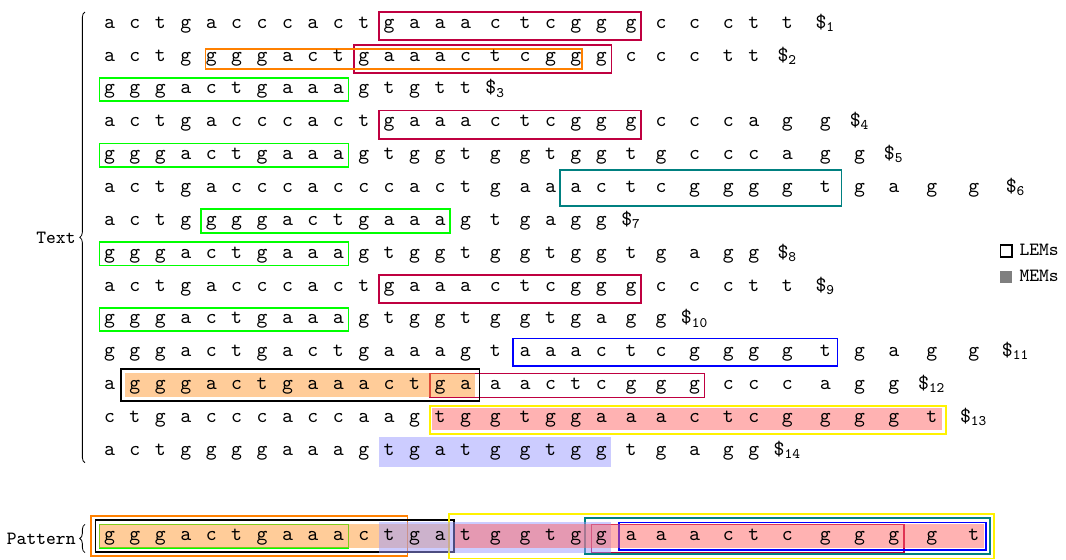}
    \caption{MEMs and LEMs of a pattern (haplotype) vs a text (pangenome). Here, haplotype $i$ is the sequence of characters between $\$_{i-1}$ and $\$_i$. The text is the concatenation of the haplotypes $T=``a c t g a c c c a c t g a a a c t c g g g c c c t t \$_1a c t g g g g a c t g a a a c t c g g g c c c t t \$_2\dots"$. MEMs and long LEMs of the pattern (a haplotype not contained in the pangenome) with respect to the text (the pangenome) are highlighted. MEMs are shaded in while LEMs are boxed in. The length threshold $\mathcal{L} = 10$ for the long LEMs. In this example, MEMs are only able to detect relationships among the haplotypes most closely related to the pattern haplotype. Haplotypes similar to the pattern but not maximally similar at any location remain undetected. Notably, haplotype 2 is very similar to the pattern but doesn't contain any MEMs with it. The number of undetected similar haplotypes in biobank scale haplotype panels may be an order of magnitude larger.
    }
    \label{memslemsfig}
\end{figure}

\section{Methods}
Here we describe the main results of our paper. In~\Cref{seclcp}, we prove move data structures can compute $\phi$ and $PLCP$ in constant time. Then we describe OptBWTRL, our modification of OptBWTR that utilizes these move data structures. In~\Cref{longsect}, we describe multiple algorithms for long LEM query provided an OptBWTRL of the text and matching statistics of the pattern with respect to the text.
\subsection{Computing LCP with Move Data Structures}\label{seclcp}
We define $p^+_j$ to be the $j$-th smallest suffix that occurs at the top of a run in the BWT. Therefore let (i) $p^+_1<p^+_2<\dots<p^+_r<p^+_{r+1}=n+1$ and (ii) $\{p^+_1,p^+_2,\dots,p^+_r,p^+_{r+1}\}=\{SA[l_1],SA[l_2],\dots,SA[l_r],n+1\}$. \Cref{lcp_blocks_lemma} and its proof are phrased very similarly to Lemma 4 in~\cite{nishimoto_et_al:LIPIcs.ICALP.2021.101} to demonstrate its derivativeness and the similarity of the properties.

\begin{lemma}\label{lcp_blocks_lemma}
    (i) Let $x$ be the integer satisfying $p^+_x \leq i < p^+_{x+1}$ for some integer $i\in [1,n]$. Then $LCP[ISA[i]] = LCP[ISA[p^+_x]] - (i - p^+_x)$.
\end{lemma}
\begin{proof}
\Cref{lcp_blocks_lemma}(i) clearly holds for $i=u_{\delta'[x]}$. We show that Lemma 4(i) holds for $i\neq p^+_x$ (i.e., $i>p^+_x$). Let $s_t$ be the position in $SA$ with sa-value $p^+_x+t$ for an integer $t\in[1,y]$ (i.e., $SA[s_t]=p^+_x+t$) where $y = i - p^+_x$. Two adjacent positions $s_t - 1$ and $s_t$ are contained in an interval $[l_v, l_v + |L_v| -1]$ on $LCP$ which corresponds to the $v$-th run $L_v$ of $L$. This is because $s_t$ is not the starting position of a run, i.e., $(SA[s_t]=p^+_x+t)\notin \{p^+_1,p^+_2,\dots,p^+_r\}$. The LF function maps $s_t$ to $s_{t-1}$, where $s_0$ is the position with sa-value $p^+_x$. LF also maps $s_t - 1$ to $s_{t-1} - 1$ by Lemma 3(i) of~\cite{nishimoto_et_al:LIPIcs.ICALP.2021.101}. $LCP[s_{t-1}] = LCP[s_t] + 1$ due to $s_t$ and $s_t - 1$ being in the same interval on $L$, $L_v$. These relationships produce $y$ equalities $LCP[s_0] = LCP[s_1] + 1, LCP[s_1] = LCP[s_2]+1,\dots,LCP[s_{y-1}] = LCP[s_y] + 1$. The equalities lead to $LCP[s_0] = LCP[s_y] + y$, and therefore $LCP[s_y] = LCP[s_0] - y$. Which represents $LCP[ISA[i]] = LCP[ISA[p^+_x] - (i-p^+_x)$ by $ISA[i] = s_y, ISA[p^+_x]= s_0,$ and $y = (i - p^+_x)$.
\end{proof}

\begin{lemma}
    (i) Let $x$ be the integer satisfying $p^+_x \leq i < p^+_{x+1}$ for some integer $i\in[1,n]$. Then $PLCP[i] = PLCP[p^+_x] - (i - p^+_x)$.
\end{lemma}
\begin{proof}
By \Cref{lcp_blocks_lemma} and $PLCP[j] = LCP[ISA[j]]$ for all $j\in [1,n]$~\cite{10.1007/978-3-642-02441-2_17}.    
\end{proof}

\subsubsection{Move Data Structure for \texorpdfstring{$\phi$}{\textit{φ}}}
$p^+_j$ remains as defined in the previous section. Let $\delta^+$ be a permutation of $[1,r]$ satisfying $\phi(p^+_{\delta^+[1]}) < \phi(p^+_{\delta^+[2]}) < \dots < \phi(p^+_{\delta^+[r]})$. $\phi$ has the following properties on RLBWT.
\begin{lemma}\label{philemma}
    The following three statements hold: (i) Let x be the integer satisfying $p^+_x \leq i < p^+_{x+1}$ for some integer $i\in[1,n]$. Then $\phi(i) = \phi(p^+_x) + (i-p^+_x)$; (ii) $\phi(p^+_{\delta^+[1]}) = 1$ and $\phi(p^+_{\delta^+[i]}) = \phi(p^+_{\delta^+[i-1]}) + d$ where $d = p^+_{\delta^+[i-1]+1} - p^+_{\delta^+[i-1]}$; (iii) $p^+_1 = 1$.
\end{lemma}

\begin{proof}
(i) \Cref{philemma}(i) clearly holds for $i=p_x$. We show that \Cref{philemma}(i) holds for $i\neq p^+_x$ (i.e., $i>p^+_x$). Let $s_t$ be the position in $SA$ with sa-value $p^+_x+t$ for an integer $t\in[1,y]$ (i.e., $SA[s_t] = p^+_x+t$), where $y = i - p^+_x$. Two adjacent positions $s_t$ and $s_t - 1$ are contained in an interval $[l_v,l_v+|L_v| - 1]$ on SA (i.e., $s_t,s_t-1\in[l_v,l_v+|L_v|-1]$), which corresponds to the $v$-th run $L_v$ of $L$. This is because $s_t$ is not the starting position of a run, i.e. $(SA[s_t]=p^+_x+t)\notin \{p^+_1,p^+_2,\dots,p^+_r\}$. The LF function maps $s_t$ to $s_{t-1}$, where $s_0$ is the position with sa-value $v_x$. LF also maps $s_t-1$ to $s_{t-1}-1$, by Lemma 3(i) of \cite{nishimoto_et_al:LIPIcs.ICALP.2021.101}. The two mapping relationships established by LF produce $y$ equalities $\phi(SA[s_1]) = \phi(SA[s_0]) + 1, \phi(SA[s_2]) = \phi(SA[s_1]) +1,\dots,\phi(SA[s_y])=\phi(SA[s_{y-1}])+1$. The equalities lead to $\phi(SA[s_y]) = \phi(SA[s_0]) + y$, which represents $\phi(i) = \phi(p^+_x) + (i-p^+_x)$ by $SA[s_y] = i, SA[s_0] = p^+_x$, and $y = i - p^+_x$. 

(ii) Let $p$ be the integer satisfying $L_p=\$$. Then there exists an integer $q$ such that $p^+_q$ is the sa-value at position $l_p+1$ ($l_p+1 = l_{p+1}$) if $p \neq r$; otherwise if $p=r$, $p^+_q$ is the sa-value at position 1 and $q=1$. $\phi(p^+_q) = 1$, because $SA[l_p] = 1$ always holds. Hence $\phi(p^+_{\delta^+[1]}) = 1$ holds by $\delta^+[1] = q$.

Next, $\phi(p^+_{\delta^+[i]}) = \phi(p^+_{\delta^+[i-1]}) + d$ holds for any $i\in[2,r]$ because (a) $\phi$ maps the interval $[p^+_{\delta^+[i]},p^+_{\delta^+[i]} + d-1]$ into the interval $[\phi(p^+_{\delta^+[i]}), \phi(p^+_{\delta^+[i]}) + d - 1]$ by \Cref{philemma}(i) for any $i\in[1,r]$, (b) $\phi$ is a bijection from $[1,n]$ to $[1,n]$, and (c) $\phi(p^+_{\delta^+[1]}) < \phi(p^+_{\delta^+[2]}) < \dots < \phi(p^+_{\delta^+[r]})$ holds.

(iii) Recall that $p$ is an integer satisfying $L_p = \$$. Then there exists an integer $q'$ such that $p^+_{q'}$ is the sa-value at position $l_p$. Finally, recall $SA[l_p]=1$. Hence, $v_1 = v_{q'} = 1$ holds.
\end{proof}

We can compute $\phi$ by using a move data structure. A sequence $I_\phi$ consists of $r$ pairs $(p^+_1, \phi(p^+_1)), (p^+_2, \phi(p^+_2)),\dots,(p^+_r,\phi(p^+_r))$. $I_\phi$ satisfies the three conditions of a disjoint interval sequence by \Cref{philemma}, and $\phi$ is equal to the bijective function represented by $I_\phi$. 

\begin{lemma}\label{phifunclemma}
    (i) $I_\phi$ is a disjoint interval sequence. (ii) $\phi$ is equal to the bijective function represented by $I_\phi$.
\end{lemma}
\begin{proof}
(i) $I_\phi$ has the following three properties: (a) $p^+_1 = 1 < p^+_2 < \dots < p^+_r \leq n$ holds by \Cref{philemma}(iii) and the definition of the sequence $p^+_1,p^+_2,\dots,p^+_{r+1}$, (b) $\phi(p^+_{\delta^+[1]}) = 1$ by \Cref{philemma}(ii), and (c) $\phi(p^+_{\delta^+[i]}) = \phi(p^+_{\delta^+[i-1]}) + (p^+_{\delta^+[i-1]+1} - p^+_{\delta^+[i-1]})$. Therefore $I_\phi$ satisfies the three conditions of the disjoint interval sequence.

(ii) Let $f_\phi$ be the bijective function represented by $I_\phi$. Then $f_\phi(i) = \phi(p^+_x) + (i-p^+_x)$ where $x$ is the integer such that $p^+_x \leq i < p^+_{x+1}$ holds. On the other hand, $\phi(i) = \phi(p^+_x) + (i - p^+_x)$ holds by \Cref{philemma}(i). Therefore $f_\phi(i) = \phi(i)$ and $f_\phi$ and $\phi$ are the same function.
\end{proof}

Let $F(I_\phi)$ be the move data structure built on the balanced interval sequence $B(I_\phi)$ for $I_\phi$. By Lemma 6 of \cite{nishimoto_et_al:LIPIcs.ICALP.2021.101}, $F(I_\phi)$ requires $O(r)$ words of space. By the results of Section 3.2 of \cite{nishimoto_et_al:LIPIcs.ICALP.2021.101}, evaluation of a move query using a move data structure for a balanced disjoint interval sequence takes constant time. Finally, $\phi(i) = i'$ holds for a move query $Move(B(I_\phi),i,x) = (i',x')$ by \Cref{phifunclemma}. Therefore we have proved (i) of the following lemma.
\begin{lemma}\label{plcpabovelemmadatastructure}
    (i) There exists a move data structure $F(I_\phi)$ that computes $\phi(i)$ in $O(r)$ space and constant time given $x$, the index of the input interval of $I_\phi$ that contains $i$. (ii) This move data structure can be modified to also compute $PLCP[i]$ in $O(r)$ space and constant time given $i$ and $x$. Call the modified move data structure $F(I_{\phi,PLCP})$. 
\end{lemma}
\begin{proof}
    Say that $B(I_\phi)$ has $k^+$ input intervals and the $i$-th input interval is $[p^+_i,p^+_{i+1}]$. Then we modify the move data structure $F(I_\phi)$ by adding an array $LCP^+$ of size $k^+$. $LCP^+[i]$ stores the value $LCP[ISA[p^+_x]]$ for each $x\in[1,k^+]$. ($PLCP[p^+_x] = LCP[ISA[p^+_x]]$.) $PLCP[i] = PLCP[p^+_x] - (i - p^+_x)$ by \Cref{lcp_blocks_lemma}. Therefore $PLCP[i]$ is computed in constant time by evaluating $LCP^+[x] - (i - p^+_x)$. Call this modified move data structure $F(I_{\phi,PLCP})$.
\end{proof}

A similar function that we may need to compute is $LCP[i+1]$ given $SA[i]$, i.e. given $SA[i] = y$, compute $|lcp(T[y,n], T[\phi^{-1}(y),n])| = PLCP[\phi^{-1}(y)] = LCP[i+1]$. Nishimoto and Tabei described $F(I_{SA})$, a move data structure computing $\phi^{-1}$. $F(I_{SA})$ can be modified to compute $PLCP[\phi^{-1}(y)]$ in constant time as well in a similar fashion to the modification of the $F(I_\phi)$ data structure. Call the disjoint interval sequence $F(I_{SA})$ is built on $B(I_{SA})$. Call the $i$-th input interval of $B(I_{SA})$ $[p^-_i,p^-_{i+1}-1]$, where $B(I_{SA})$ has $k^-$ input intervals and $p^-_{k^-+1} = n+1$. Note that by the construction of Nishimoto and Tabei, every suffix at the bottom of a BWT run is the start of an input interval, $\{SA[l_2-1],SA[l_3-1],SA[l_4-1],\dots,SA[l_r-1],SA[n]\}\subseteq \{p^-_1,p^-_2,\dots,p^-_k\}$. Then for any $i,j\in[p^-_x,p^-_{x+1}-1], \phi^{-1}(i) - \phi^{-1}(j) = i - j$. Therefore $\phi^{-1}(i) = \phi^{-1}(j) + (i-j)$ (see Lemma 4 in \cite{nishimoto_et_al:LIPIcs.ICALP.2021.101}). Below, we prove (ii) that for any $i,j\in[p^-_x,p^-_{x+1}-1], PLCP[\phi^{-1}(i)] - PLCP[\phi^{-1}(j)] = j - i$, therefore $PLCP[\phi^{-1}(i)] = PLCP[\phi^{-1}(j)] + j - i$.

\begin{lemma}\label{PLCPBelowLemma}
    Let $x$ be the integer satisfying $p^-_x \leq i < p^-_{x+1}$ for some $i\in[1,n]$. Then (i) $PLCP[\phi^{-1}(i)] = PLCP[\phi^{-1}(p^-_x)] - (i - p^-_x)$. Therefore, (ii) for any $i,j\in[p^-_x,p^-_{x+1}-1]$, $PLCP[\phi^{-1}(i)] = PLCP[\phi^{-1}(p^-_x)] - (i - p^-_x)$, $PLCP[\phi^{-1}(j)] = PLCP[\phi^{-1}(p^-_x)] - (j - p^-_x)$, and $PLCP[\phi^{-1}(i)] - PLCP[\phi^{-1}(j)] = j - i$.
\end{lemma}
\begin{proof}
    \Cref{PLCPBelowLemma}(i) clearly holds for $i=p^-_x$. We show that \Cref{PLCPBelowLemma}(i) holds for $p^-_x < i < p^-_{x+1}$ (i.e. $i\neq p^-_x$). Let $s_t$ be the position in the SA with sa-value $p^-_x+t$ for an integer $t\in[1,y]$ where $y=i-p^-_x$. Two adjacent positions $s_t$ and $s_t+1$ are contained in an interval $[l_v, l_{v+1}-1]$ corresponding to the $v$-th run in the BWT ($L_v$). This is because $s_t$ is not the ending position of a run, $SA[s_t]\notin \{p^-_1,p^-_2,\dots,p^-_{k^-}\}$. The LF function maps $s_t$ to $s_{t-1}$, where $s_0$ is the position in the SA with value $p^-_x$. LF also maps $s_t+1$ to $s_{t-1}+1$ by Lemma 3(i) of \cite{nishimoto_et_al:LIPIcs.ICALP.2021.101}. $PLCP[\phi^{-1}(p^-_x+t-1)] = LCP[s_{t-1}+1] = LCP[s_t+1] + 1 = PLCP[\phi^{-1}(p^-_x+t)] + 1$ since $s_t$ and $s_t+1$ are in the same interval in the BWT, $L_v$. These relationships produce $y$ equalities $PLCP[\phi^{-1}(p^-_x)] = PLCP[\phi^{-1}(p^-_x+1)] + 1$,$PLCP[\phi^{-1}(p^-_x+1)] = PLCP[\phi^{-1}(p^-_x+2)] + 1$,\dots,$PLCP[\phi^{-1}(p^-_x+y-1)] = PLCP[\phi^{-1}(p^-_x+y)] + 1$. This leads to $PLCP[\phi^{-1}(p^-_x)] = PLCP[\phi^{-1}(p^-_x+y)] + y$. Which leads to $PLCP[\phi^{-1}(i)] = PLCP[\phi^{-1}(p^-_x)] - (i - p^-_x)$ by $y = i-p^-_x$ and $p^-_x+y = i$. 
\end{proof}

Therefore, the move data structure that computes $\phi^{-1}(i)$, $F(I_{SA})$, can be modified to compute $PLCP[\phi^{-1}(i)]$ as well.

\begin{lemma}\label{plcpbelowmovedatastructure}
    $F(I_{SA})$ can be modified to compute $PLCP[\phi^{-1}(i)]$ as well as $\phi^{-1}(i)$ in constant time and $O(r)$ space given $x$, the index of the input interval of $B(I_{SA})$ that contains $i$. Call the modified move data structure $F(I_{\phi^{-1},PLCP})$. 
\end{lemma}
\begin{proof}
    We modify the $F(I_{SA})$ move data structure by $LCP^-$, an array of size $k^-$ where the $x$-th element stores the value $PLCP[\phi^{-1}(p^-_x)]$. Then, $PLCP[\phi^{-1}(i)]$ can be computed in constant time by evaluating $LCP^-[x] - (i - p^-_x)$ by $LCP^-[x] = PLCP[p^-_x]$ and \Cref{PLCPBelowLemma}(i). We call this modified move data structure $F(I_{\phi^{-1},PLCP})$.
\end{proof}

\subsubsection{OptBWTRL}
We slightly modify OptBWTR by adding a move data structure that computes $\phi$ and $PLCP$ and arrays that allow jumping to the closest input intervals corresponding to adjacent runs in the BWT in constant time. We call it \textit{OptBWTRL}, \textit{L} for LCP and $\mathcal{L}$ long LEMs. In addition to the structures of OptBWTR, OptBWTRL contains $F(I_{\phi,PLCP})$, $ND$, $PD$, $SA^-$, $SA^+_{\phi}$, $SA^-_{index}$, and $SA^-_{\phi}$. Furthermore, the $F(I_{SA})$ move data structure of OptBWTR is replaced by the $F(I_{\phi^{-1},PLCP})$ move data structure described in \Cref{plcpbelowmovedatastructure}. Recall that $B(I_{LF})$ is the disjoint interval sequence the move data structure $F(I_{LF})$ is built on. Let $B(I_{LF})$ contain $k$ input intervals where the $i$-th input interval is $[p_i,p_{i+1}-1]$, and $p_{k+1} = n+1$. Further recall that every input interval is contained in a run in the BWT, i.e. for all $i\in[1,k]$, $\forall j,j'\in [p_i,p_{i+1}-1], L[j] = L[j']$. Then, $ND$ and $PD$ are arrays of length $k$ where $ND$ contains the index of the next input interval with a different character in the BWT and $PD$ contains the index of the previous input interval with a different character in the BWT. Formally, for all $i\in[1,k]$, $ND[i] = \min\{j>i | L[p_i] \neq L[p_j]\}$, and $PD[i] = \max\{j<i | L[p_i]\neq L[p_j]\}$. If no such $j$ exists, $ND[i] = k + 1$ and $PD[i] = -1$. $ND$ and $PD$ can be constructed in $O(k)$ (and therefore, $O(r)$) time given $L_{first}$. $SA^-$ are samples of the $SA$ at the ends of input intervals of $B(I_{LF})$. $SA^+_{\phi}$ are the indices of the input intervals of the top of $B(I_{LF})$ input interval suffix array samples in $F(I_{\phi,PLCP})$. $SA^-_{index}$ and $SA^-_{\phi}$ are the indices of the input intervals of the bottom of $B(I_{LF})$ input interval suffix array samples in $F(I_{\phi^{-1},PLCP})$ and $F(I_{\phi,PLCP})$ respectively. Below, let $[p^+_i, p^+_{i+1}-1]$ and $[p^-_i, p^-_{i+1}-1]$ be the $i$-th input intervals of $F(I_{\phi,PLCP})$ and $F(I_{\phi^{-1},PLCP})$ respectively. Then, OptBWTRL differs from OptBWTR in the following ways.
\begin{itemize}
    \item Replaced $F(I_{SA})$ with $F(I_{\phi^{-1},PLCP})$ from \Cref{plcpbelowmovedatastructure}.
    \item Added $F(I_{\phi,PLCP})$ from \Cref{plcpabovelemmadatastructure}.
    \item Added $ND$ and $PD$. $ND[i] = \min_{j>i}\{L_{first}[j]\neq L_{first}[i]$ or $j = n+1\}$. $PD[i] = \max_{j>i}\{L_{first}[j]\neq L_{first}[i]$ or $j = -1\}$.
    \item Added $SA^-$. $SA^-[i] = SA[l_{i+1}-1]$.
    \item Added $SA^+_{\phi}$. $SA^+_\phi[i] = j$ s.t. $p^+_j\leq SA^+[i] < p^+_{j+1}$.
    \item Added $SA^-_{index}$. $SA^-_{index}[i] = j$ s.t. $p^-_j \leq SA^-[i] < p^-_{j+1}$.
    \item Added $SA^-_{\phi}$. $SA^-_{\phi}[i] = j$ s.t. $p^+_j\leq SA^-[i] < p^+_{j+1}$.
\end{itemize}

\subsection{Computing Long LEMs}\label{longsect}
Here, we describe an algorithm for outputting all the long LEMs of a pattern $P$ with respect to a text $T$ in $O(m+occ)$ expected time using an index of size $O(r)$ words given the matching statistics of $P$ with respect to $T$ and an OptBWTRL of $T$. $m$ is the length of $P$ and $occ$ is the number of long LEMs $P$ has with $T$. Furthermore, the matching statistics are slightly augmented to contain the input intervals it's corresponding data are contained in. In particular, the input interval of $F(I_{LF})$ that $MS.row$ is contained in is stored as $MS.i$, the input interval of $F(I_{\phi,PLCP})$ that $MS.suff$ is contained in is stored as $MS.w$, and the input interval of $F(I_{\phi^{-1},PLCP})$ that $MS.suff$ is contained in is stored as $MS.x$. Note that the long LEM query algorithm we present here does not necessarily result in an $O(m+occ)$ expected time algorithm for outputting all long LEMs of $P$ with respect to $T$ given a OptBWTRL of $T$ because an algorithm for computing the matching statistics of $P$ with respect to $T$ in $O(m)$ time and $O(r)$ space is not known. 

We define the \textit{balanced sa\textsubscript{lcp}-interval} of a string $P$ as a 13-tuple $(b,d,e,SA[b],SA[d],SA[e],\allowbreak i,j,k,v,w,x,y)$ where $[b,e]$ is the sa-interval of $P$, $d\in[b,e]$, $i,j,$ and $k$ are the indexes of the input intervals of $B(I_{LF})$ that contain $b,d,$ and $e$ respectively, $v$ and $w$ are the indexes of the input intervals of $B(I_{\phi,PLCP})$ containing $SA[b]$ and $SA[d]$ respectively, and $x$ and $y$ are the indexes of the input intervals of $B(I_{\phi^{-1},PLCP})$ of $SA[d]$ and $SA[e]$ respectively. The balanced sa\textsubscript{lcp}-interval keeps track of three positions in the sa-interval. The top ($b$), bottom ($e$), and the middle ($d$). $d$ is any position in the interval, it may be equivalent to the top or the bottom. Each position also maintains its corresponding suffix array value and index of the input interval of the position in $F(I_{LF})$ ($i,j,$ and $k$ for top, middle, and bottom respectively). Finally, the top maintains the index of the input interval of its sa-value in $F(I_{\phi,PLCP})$ ($v$), the bottom maintains the index of the input interval of its sa-value in $F(I_{\phi^{-1},PLCP})$ ($y$), and the middle maintains the index of the input interval of its sa-value in both $F(I_{\phi,PLCP})$ and $F(I_{\phi^{-1},PLCP})$ ($w$ and $x$ respectively). The balanced sa\textsubscript{lcp}-interval of a string $P$ with no occurrences in $T$ is undefined. 

The high level idea of the long LEM finding algorithm is to compute the balanced sa\textsubscript{lcp}-interval of adjacent substrings of length $\mathcal{L}$ of the pattern while outputting long LEMs along the way. I.E. given the balanced sa\textsubscript{lcp}-interval of $P[f+1, f+\mathcal{L}]$, compute the sa\textsubscript{lcp}-interval of $P[f,f+\mathcal{L}-1]$ and output all long LEMs of the form $P[f+1,g] = T[f',g']$. We call this problem \textit{long sa\textsubscript{lcp}-interval advancement}. Given an algorithm for long sa\textsubscript{lcp}-interval advancement in $O(t_\mathcal{L})$ time, a straightforward long LEM computation algorithm is iterating from $f=m\to 1$, repeatedly advancing the sa\textsubscript{lcp}-interval and outputting all long LEMs in $O(mt_\mathcal{L})$ time. In \Cref{sa-lcp-extension-sect}, we outline an algorithm for balanced sa\textsubscript{lcp}-interval extension and in \Cref{long-sa-lcp-advancement-sect}, we outline an algorithm for long sa\textsubscript{lcp}-interval advancement. These algorithms result in an $O(m+occ)$ expected time algorithm for long LEM computation.

\subsubsection{Balanced \texorpdfstring{sa\textsubscript{lcp}}{sa\_\{lcp\}}-interval Extension}\label{sa-lcp-extension-sect}
Here, we provide algorithms for obtaining the balanced sa\textsubscript{lcp}-interval of $cP$ given the balanced sa\textsubscript{lcp}-interval of $P$ and an OptBWTRL of $T$. The first algorithm runs in $O(\log\log_w\sigma)$ time by making use of the rank-select structure on $L_{first}$. The second runs in time linear to the number of runs in the balanced sa\textsubscript{lcp}-interval of $P$, $r_P$, by iterating through them. Call the balanced sa\textsubscript{lcp}-interval of $P$ $(b,d,e,SA[b],SA[d],SA[e],i,j,k,v,w,x,y)$ and the balanced sa\textsubscript{lcp}-interval of $cP$ $(b',d',e',SA[b'], SA[d'], SA[e'], i', j', k', v',w',x',y')$. Recall that $p_j$ is the starting index of the $j$-th input interval of $F(I_{LF})$, $p^+_j$ is the starting index of the $j$-th input interval of $F(I_{\phi,PLCP})$, and $p^-_j$ is the starting index of the $j$-the input interval of $F(I_{\phi^{-1},PLCP})$.

We first discuss the computation of the top values, $b',SA[b'],i',$ and $v'$. If $L[b] = c$, then $b' = LF[b]$ and $i'$ can be computed with $F(I_{LF})$ in constant time using $(b,i)$. $SA[b'] = SA[b] - 1$, and $v'=v$ if $SA[b]\neq p^+_v$, otherwise $v' = v - 1$. If $L[b]\neq c$, $b' = LF[\hat{b}]$, where $\hat{b}$ is the first location in $[b,e]$ such that $L[\hat{b}] = c$. If $\hat{i}$ is the index of first input interval $i \leq \hat{i} \leq k$ such that $L_{first}[\hat{i}] = c$, then $\hat{b} = p_{\hat{i}}$, where $p_a$ is the starting position of the $a$-th input interval of $F(I_{LF})$. $\hat{i}$ can be computed in $O(\log\log_w\sigma)$ time using $R(L_{first})$ or $O(r_P)$ time by iterating through the runs of balanced sa\textsubscript{lcp}-interval of $P$ using the $ND$ array. Then, $i'$ and $b' = LF[\hat{b}]$ can be computed with $F(I_{LF})$ in constant time using $(\hat{b}, \hat{i})$. $SA[b'] = SA^+[\hat{i}]-1$, and $v' =SA^+_\phi[\hat{i}]-1$. 

The bottom values $e', SA[e'], k'$, and $y'$ can be computed in a similar fashion. If $L[e] = c$, then $e' = LF[e]$ and $k'$ can be computed with $F(I_{LF})$ in constant time using $(e,k)$. $SA[e'] = SA[e] - 1$, and $y'=y$ if $SA[e]\neq p^-_k$, otherwise $y'=y-1$. If $L[e]\neq c$, then $e'=LF[\hat{e}]$, where $\hat{e}$ is the last location in $[b,e]$ such that $L[\hat{e}] = c$. If $\hat{k}$ is the index of the last input input interval $i \leq \hat{k} \leq k$ such that $L_{first}[\hat{k}] = c$, then $\hat{e} = p_{\hat{k}+1}-1$. $\hat{k}$ can be computed in $O(\log\log_w\sigma)$ time using $R(L_{first})$ or $O(r_p)$ time by iterating through the runs of the balanced sa\textsubscript{lcp}-interval of $P$ using the $PD$ array. Then, $k'$ and $e' = LF[\hat{e}]$ can be computed with $F(I_{LF})$ in constant time using $(\hat{e}, \hat{k})$. Finally, $SA[e'] = SA^-[\hat{k}]-1$, and $y'=SA^+_{index}[\hat{k}]-1$.

Lastly, the middle values $d', SA[d'], j', w'$ and $x'$ need to be computed. If $L[d] = c$, then $d'=LF[d]$ and $j'$ can be computed in constant time with $F(I_{LF})$ using $(d,j)$. $SA[d'] = SA[d-1]$. $w'=w$ if $SA[d]\neq p^+_w$, otherwise $w'=w-1$. Finally, $x'=x$ if $SA[d]\neq p^-_x$, otherwise $x'=x-1$. If $L[d]\neq c$ and $cP$ occurs in $T$, then there is a preceding or succeeding input interval of $B(I_{LF})$ that intersects with the balanced sa\textsubscript{lcp}-interval of $P$ and has value $c$ in the BWT. Suppose there is a preceding interval, $\hat{j}$. Then the middle values can be updated similar to the bottom values. Let $\hat{d} = p_{\hat{j}+1}-1$, then $j'$ and $d' = LF[\hat{d}]$ are computed in constant time with $F(I_{LF})$, $SA[d'] = SA^-[\hat{j}]-1$, $x' = SA^-_{index}[\hat{j}]-1$, and $w' = SA^-_{\phi}[\hat{j}]$ if $SA[\hat{d}] \neq p^+_{SA^-_{\phi}[\hat{j}]}$, otherwise $w' = SA^-_{\phi}[\hat{j}] - 1$. If there is no preceding interval, then set $\hat{j}$ to the index of the succeeding interval. Then the middle values can be updated similar to the top values. Let $\hat{d}=p_{\hat{j}}$, then $j'$ and $d'=LF[\hat{d}]$ are computed in constant time with $F(I_{LF})$, $SA[d'] = SA^+[\hat{j}]$, $w' = SA^+_{\phi}[\hat{j}]-1$, and $x' = SA^+_{index}[\hat{j}]$ if $SA[\hat{d}] \neq p^-_{SA^+_{index}[\hat{j}]}$, otherwise $x' = SA^+_{index}[\hat{j}] - 1$. The index, $\hat{j}$, of the preceding or succeeding interval in the sa\textsubscript{lcp}-interval of $P$ with value $c$ in the BWT can be found in $O(\log\log_w\sigma)$ time with $R(L_{first})$ or $O(r_P)$ time by iterating through the runs in the BWT with $PD$ and $ND$. Therefore, the balanced sa\textsubscript{lcp}-interval of $cP$ can be computed in $O(\log\log_w\sigma)$ time or $O(r_P)$ time given the balanced sa\textsubscript{lcp}-interval of $P$.

\subsubsection{Long \texorpdfstring{sa\textsubscript{lcp}}{sa\_\{lcp\}}-interval Advancement}\label{long-sa-lcp-advancement-sect}
Let $occ_{start,f+1}$ be the number of long LEMs of the  form $P[f+1,g] = T[f',g']$ and $occ_{end,f+\mathcal{L}-1}$. be the number of long LEMs of the form $P[h, f+\mathcal{L}-1] = T[f'',g'']$. Here, we describe an algorithm that computes the balanced sa\textsubscript{lcp}-interval of $P[f,f+\mathcal{L}-1]$ and outputs all $occ_{start,f+1}$ long LEMs of the form $P[f+1,g] = T[f',g']$ in $O(occ_{start,f+1}+occ_{end,f+\mathcal{L}-1})$ expected time given the balanced sa\textsubscript{lcp}-interval of $P[f+1,f+\mathcal{L}]$, an OptBWTRL of $T$, and a dynamic dictionary of the suffixes of $T$ present in the balanced sa\textsubscript{lcp}-interval of $P[f+1,f+\mathcal{L}]$.

We begin with the description of the dynamic dictionary, $dict_{occ}$. There are numerous dynamic dictionary data structures that support expected constant time insertion, deletion, and queries~\cite{bender2022optimal,rrr10.1007/3-540-45061-0_30,bendersimpledoi:10.1137/1.9781611977936.33,demaine10.1007/11682462_34}. Therefore, we maintain a dynamic dictionary of the suffixes in the balanced sa\textsubscript{lcp}-interval. More precisely, if the balanced sa\textsubscript{lcp}-interval of $P[f+1,f+\mathcal{L}]$ is $(b,d,e,SA[b],SA[d],SA[e],i,j,k,v,w,x,y)$, then the dynamic dictionary provided as input to the long sa\textsubscript{lcp}-interval advancement algorithm has $e-b+1$ elements. $\forall a\in[b,e], SA[b] - (f + 1)$ is contained in the dictionary and has the value $(f+1) + |lcp(T[SA[a],n],P[f+1,m])|-1= f + |lcp(T[SA[a],n],P[f+1,m])|$ associated with it. I.E. the value associated with each suffix $SA[a]$ of the text contained in the dictionary is the ending position (in the pattern) of the longest match between suffix $SA]a$ of the text and suffix $f+1$ of the pattern.

Here, we describe the procedure for outputting all $occ_{start,f+1}$ long LEMs of the form $P[f+1,g]=T[f',g']$ in $O(occ_{start,f+1})$ expected time. The high level idea is to iterate through the input intervals of $B(I_{LF})$, skipping intervals corresponding to a run of $P[f]$ in constant time per run using $ND$. We outline two functions: $outputMatchesDown(s, \iota, z)$ and $outputMatchesUp(s,\iota,z)$. For both functions, $s$ represents a suffix of $T$ and $\iota$ is the index of the input interval that contains it in $F(I_{\phi^{-1},PLCP})$ and $F(I_{\phi,PLCP})$ in $outputMatchesDown$ and $outputMatchesUp$ respectively. $z$ represents the number of matches to output (directly above $s$ in $SA$ for $outputMatchesUp$ and directly below $s$ in $SA$ for $outputMatchesDown$) including $s$. $outputMatchesUp(s,\iota,1)$ outputs a match $P[f+1,g]=T[s,s+g-f]$, where $g=dict_{occ}[s-(f+1)]$, and removes the key-value pair $(s-(f+1), g)$ from $dict_{occ}$. $outputMatchesUp(s,\iota,z)$ for $z>1$ similarly outputs a match $P[f+1,g]=T[s,s+g-f]$ where $g=dict_{occ}[s-(f+1)]$, then removes the key-value pair $(s-(f+1),g)$ from $dict_{occ}$. Then, it recurses on $outputMatchesMathesUp(s',\iota',z-1)$, where $\iota'$ and $s'=\phi(s)$ are computed in constant time using $F(I_{\phi,PLCP})$. $outputMatchesDown(s,\iota,z)$ operates in the same way as $outputMatchesDown$ except it computes $\phi^{-1}$ instead of $\phi$ (using $F(I_{\phi^{-1},PLCP})$ instead of $F(I_{\phi,PLCP})$). It is simple to see that $outputMatchesUp(s,\iota,z)$ and $outputMatchesDown(s,\iota,z)$ operate in $O(z)$ expected time and output $z$ matches each. Now we utilize $outputMatchesUp$ and $outputMatchesDown$ to output the $occ_{start,f+1}$ long LEMs of the form $P[f+1,g]=T[f',g']$. If the sa\textsubscript{lcp}-interval of $P[f+1,f+\mathcal{L}]$ is fully contained in one input interval of $F(I_{LF})$, then $i=k$. If $L_{first}[i]=P[f]$, then there are no matches to output, otherwise, every suffix in the balanced sa\textsubscript{lcp}-interval needs to be outputted and we do so by calling $outputMatchesUp(SA[d],w,d-b+1)$ and $outputMatchesDown(\phi^{-1}(SA[d]),x',e-b)$, where $x'$ and $\phi^{-1}(SA[d])$ are computed with $(SA[d], x)$ and $F(I_{\phi^{-1},PLCP})$. In the case where the balanced sa\textsubscript{lcp}-interval of $P[f+1,f+\mathcal{L}]$ is not fully contained in one input interval $(i\neq j)$, we do the following. For the first input interval, $i$, if $L_{first}[i]\neq P[f]$, then the $p_{i+1}-b$ long LEMs starting at $SA[p_{i+1}-1],SA[p_{i+1}-2],\dots,SA[b]$ in the text are outputted in $O(p_{i+1}-b)$ expected time by calling $outputMatchesUp(SA^-[i],SA^-_{\phi}[i],p_{i+1}-b)$. For any middle input interval $o$, $i<o<j$, if $L_{first}[o]=P[f]$, then this run in the BWT is skipped, $o=ND[o]$. Otherwise, if $L_{first}[o]\neq P[f]$, then the long LEMs starting at $SA[p_{o}],SA[p_{o}+1],\dots,SA[p_{o+1}-1]$ are outputted by calling $outputMatchesDown(SA^+[o],SA^+_{index}[o],p_{o+1}-p_o)$. For the last input interval, $k$, if $L_{first}[k]\neq P[f]$, then the $e-p_k+1$ long LEMs starting at $SA[p_k],SA[p_k+1],\dots,SA[e]$ are outputting by calling $outputMatchesDown(SA^+[k], SA^+_{index}[k],e-p_{k}+1)$. Overall, outputting the $occ_{start,f+1}$ long LEMs of the form $P[f+1,g]=T[f',g']$ takes $O(occ_{start,f+1}+r_{P[f+1,f+\mathcal{L}]})$ expected time. Furthermore, for every run of character $P[f]$ intersecting the sa\textsubscript{lcp}-interval of $P[f+1,f+\mathcal{L}]$ except the first one, there is a run of characters $\neq P[f]$. Therefore $r_{P[f+1,f+\mathcal{L}]} = O(occ_{start,f+1})$ Therefore outputting the $occ_{start,f+1}$ long LEMs of the form $P[f+1,g]=T[f',g']$ takes $O(occ_{start,f+1})$ expected time.

Finally, we must compute the balanced sa\textsubscript{lcp}-interval of $P[f,f+\mathcal{L}-1]$. First suppose that the balanced sa\textsubscript{lcp}-interval of $P[f+1,f+\mathcal{L}]$ is nonempty. Then, we use the algorithm described in \Cref{sa-lcp-extension-sect} to obtain the sa\textsubscript{lcp}-interval of $P[f,f+\mathcal{L}]$ in $O(r_{P[f+1,f+\mathcal{L}]})$ time. Now, let the sa\textsubscript{lcp}-interval of $P[f,f+\mathcal{L}]$ be $(\hat{b},\hat{d},\hat{e},SA[\hat{b}],SA[\hat{d}],SA[\hat{e}],\hat{i},\hat{j},\hat{k},\hat{v},\hat{w},\hat{x},\hat{y})$ and the sa\textsubscript{lcp}-interval of $P[f,f+\mathcal{L}-1]$ be $(b',d',e',SA[b'], SA[d'], SA[e'], i', j', k', v',w',x',y')$. These sa\textsubscript{lcp}-intervals differ only by those suffixes of the text whose $lcp$ with $P[f,m]$ has length exactly $\mathcal{L}$. There are exactly $occ_{end,f+\mathcal{L}-1}$ such suffixes. Furthermore, $PLCP[b'] < \mathcal{L}$ and $PLCP[\phi^{-1}(e')] < \mathcal{L}$. Finally, $\forall b'<a\leq \hat{b}$, $LCP[a] = PLCP[SA[a]] \geq \mathcal{L}$, and $\forall \hat{e}\leq a < e'$, $LCP[a+1] = PLCP[SA[a+1]] = PLCP[\phi^{-1}(a)] \geq \mathcal{L}$. Therefore, we initialize $b'= \hat{b}$, $SA[b']=SA[\hat{b}]$, $i'=\hat{i}$, and $v'=\hat{v}$. Then, while $LCP[b'] = PLCP[SA[b']] \geq \mathcal{L}$, we (i) set $i' = i' -1$ if $b' = p_{i'}$, (ii) set $b' = b'-1$, (iii) update $SA[b']$ and $v'$ by $F(I_{\phi,PLCP})$, and (iv) insert the key $SA[b']-f$ into $dict_{occ}$ with value $f+\mathcal{L}-1$. When $LCP[b'] = PLCP[SA[b']] < \mathcal{L}$, the final value $b'$ has been computed. Similarly for $e'$, we initialize $e'=\hat{e}, SA[e'] = SA[\hat{e}], k' = \hat{k},$ and $y'=\hat{y}$. Then, while $LCP[e'+1] = PLCP[SA[e'+1]] = PLCP[\phi(SA[e'])] \geq \mathcal{L}$, we (i) set $k' = k'-1$ if $e'=p_{k'+1}-1$, (ii) set $e'=e'-1$, (iii) update $SA[e']$ and $y'$ by $F(I_{\phi^{-1},PLCP})$, and (iv) insert the key $SA[e']-f$ into $dict_{occ}$ with value $f+\mathcal{L}-1$. When $LCP[e'+1] = PLCP[SA[e'+1]] = PLCP[\phi^{-1}(e')] < \mathcal{L}$, the final value $e'$ has been computed. Overall this takes constant time per suffix added to the interval, therefore $O(occ_{end,f+\mathcal{L}-1})$ time. 

If the balanced sa\textsubscript{lcp}-interval of $P[f+1,f+\mathcal{L}]$ is empty, the balanced sa\textsubscript{lcp}-interval of $P[f,f+\mathcal{L}-1]$ is only nonempty if $MS[f].len = \mathcal{L}$. If it is, we initialize the balanced sa\textsubscript{lcp}-interval of $P[f,f+\mathcal{L}-1]$ to $(\hat{b} = MS[f].row,\hat{d} = MS[f].row,\hat{e}=MS[f].row,SA[\hat{b}]=MS[f].suff,SA[\hat{d}]=MS[f].suff,SA[\hat{e}]=MS[f].suff,\hat{i}=MS[f].i,\hat{j}=MS[f].i,\hat{k}=MS[f].i,\hat{v}=MS[f].w,\hat{w}=MS[f].w,\hat{x}=MS[f].x,\hat{y}=MS[f].x)$ and insert the key $MS.suff - f$ into $dict_{occ}$ with value $f+\mathcal{L}-1$. Then, the interval is expanded to the sa\textsubscript{lcp}-interval of $P[f,f+\mathcal{L}-1]$ in $O(occ_{end,f+\mathcal{L}-1})$ time as in the other case. 

In the case where the balanced sa\textsubscript{lcp}-interval of $P[f+1,f+\mathcal{L}]$ is empty, long sa\textsubscript{lcp}-interval advancement is performed in $O(occ_{end,f+mathcal{L}-1})$ expected time . If it is not empty, the algorithm we have described first performs sa\textsubscript{lcp}-interval extension, obtaining the sa\textsubscript{lcp}-interval of $P[f,f+\mathcal{L}]$ in $O(r_{P[f+1,f+\mathcal{L}]})$ time and then takes $O(occ_{end,f+mathcal{L}-1})$ expected time to compute the sa\textsubscript{lcp}-interval of $P[f,f+\mathcal{L}-1]$ from the sa\textsubscript{lcp}-interval of $P[f,f+\mathcal{L}]$. Finally, $r_{P[f+1,f+\mathcal{L}]}=O(occ_{end,f+mathcal{L}-1})$. Therefore, the algorithm described here performs the long sa\textsubscript{lcp}-interval advancement in $O(occ_{start,f+1} + occ_{end,f+\mathcal{L}-1})$ expected time. 

\subsection{Time Complexity}
If the above algorithm is iterated from $f=m\to 0$, all long MEMs of the pattern with respect to the text are outputted. The time complexity of the algorithm is the sum of the time complexity of the $m$ long sa\textsubscript{lcp}-interval advancements. Note that the sum of $occ_{start,f+1}$ for $f=m\to 0$ is $occ$ and the sum of the $occ_{end,f+\mathcal{L}-1}$ for $f=m\to 0$ is also $occ$. Therefore, the time complexity of the algorithm overall is $O(m+occ)$ expected time. The algorithm takes $O(r)$ space for the OptBWTRL and $O(occ)$ space for maintaining the dynamic dictionary~\cite{bender2022optimal}. Also note that if a deterministic time bound is desired, this algorithm runs in $O\left(m+occ\sqrt{\frac{\log occ}{\log\log occ}}\right)$ time with the same space by replacing the dictionary with a deterministic dictionary implemented by exponential search trees~\cite{obit_10.1145/2447712.2447737,thorup_10.1145/1236457.1236460}. Recall these complexities are when given the modified matching statistics. A linear time algorithm for computing matching statistics in $O(r)$ space is not known. However, note that since the values of matching statistics are only needed for positions $i$ where $MS[i].len=\mathcal{L}$, a straightforward algorithm for long LEM query follows from our algorithm in $O(m\mathcal{L}\log\log_w\sigma+occ)$ expected time when matching statistics are not given as input. This algorithm is obtained by computing the sa\textsubscript{lcp}-interval of each $P[i,i+\mathcal{L}-1]$ independently in $O(\mathcal{L}\log\log_w\sigma)$ time using the standard count algorithm described by Nishimoto and Tabei~\cite{nishimoto_et_al:LIPIcs.ICALP.2021.101} followed by performing the long LEM query described here. Lastly, note that the long LEM query algorithm described here results in an $O(m+occ)$ expected time long LEM query algorithm in uncompressed string indexes since algorithms for $O(m)$ time matching statistics computation are known in uncompressed space.

\section{Discussion}
In this paper, we have described OptBWTRL, a modification of OptBWTR by Nishimoto and Tabei~\cite{nishimoto_et_al:LIPIcs.ICALP.2021.101}. OptBWTRL adds the ability to compute $PLCP$ and $\phi$ in constant time with additional move data structures. It also retains a space complexity of $O(r)$ words. We also define locally maximal exact matches (LEMs), a match that cannot be simultaneously extended in the pattern and the text instead of one that is only unable to be extended in the pattern (MEMs). Finally, we describe an algorithm for outputting all LEMs with length at least $\mathcal{L}$ in $O(m+occ)$ expected time given an OptBWTRL of the text and the matching statistics of the pattern with respect to the text. Note that this doesn't result in a linear time algorithm for computing long LEMs in $O(m+occ)$ expected time in $O(r)$ space because an algorithm for computing matching statistics of a pattern with respect to a text in linear time in $O(r)$ space is not known. A deterministic bound for our long LEM query algorithm is $O\left(m+occ\sqrt{\frac{\log occ}{\log\log occ}}\right)$. Finally, our long LEM query admits a direct computation of long LEMs in $O(m\mathcal{L}+occ)$ expected time without being provided matching statistics as input. This algorithm may be faster than computation of matching statistics followed by $O(m+occ)$ long LEM query in some cases, especially when $\mathcal{L}$ is small.

It is likely that the move data structures $F(I_{\phi,PLCP})$ and $F(I_{\phi^{-1},PLCP})$ can be merged into one data structure that still takes $O(r)$ space and computes $\phi,\phi^{-1},PLCP[i],$ and $PLCP[\phi^{-1}(i)]$ in constant time in one data structure. This would greatly reduce the number of samples needed per input interval $F(I_{LF})$. It would also allow bidirectional movement in the $SA$ with one input interval index. This is left as future work. Possible other future work includes a practical implementation of the structures and algorithms described here, possibly as a modification of MOVI or b-move~\cite{zakeri2024movi,bmovedepuydt_et_al:LIPIcs.WABI.2024.10}. Thirdly, the ability to compute $PLCP$ in constant time may speed up matching statistics computation in compressed space. The intuition is that when $MS[i].len \leq MS[i+1].len$, then when $MS[i].len$ is large, the sa-interval of $P[i,i+MS[i].len-1]$ is small and is faster to compute with $PLCP$ and $\phi$ and $\phi^{-1}$ than with reverse $LF$. When $MS[i].len$ is small, computing the sa-interval is faster with reverse $LF$. In Heng Li's forward-backward algorithm~\cite{liforwardbackward10.1093/bioinformatics/bts280}, the new sa-interval is always computed by reverse $LF$. Computing the sa-interval with $PLCP$ and reverse $LF$ simultaneously is likely to be faster in practice than reverse $LF$ alone while retaining the same worst case time complexity. The author's are currently exploring this idea. Furthermore, variable length threshold long LEMs may be useful. I.e. output long LEMs that are $x\%$ of the length of the MEMs in the same area. The authors believe a linear time algorithm for this or a similar problem given matching statistics exist. Finally, applications utilizing the long LEMs of a pattern with respect to the text is a possible fruitful direction for future work. Particularly in biological applications.

Long LEMs may have many biological applications. In general, in any application where long MEMs are used, long LEMs may also be used. Note that MEMs are a subset of LEMs and long MEMs are a subset of long LEMs. For example, in biobank scale haplotype datasets, long matches (long LEMs) in the PBWT have revealed genealogical relationships that set maximal matches (MEMs) are not able to uncover. As the compressive power of compressed string indexes increases and the number of variants in biobank scale whole genome sequencing data increases, storing unaligned genomes becomes more viable. In that case, algorithms for outputting long LEMs are needed to replace the long match algorithms in the PBWT. These matches have many applications from identity by descent segment detection, haplotype phasing, haplotype imputation, inferring genealogical relationships, and ancestry inference. Utilizing unaligned matches from a large collection of haplotype sequences instead of aligned matches from haplotypes aligned to a linear reference genome may also reduce reference bias. Finally, novel applications for long LEMs may exist, long LEMS may be used as seeds for seed and extend algorithms. They may be used as anchors for approximate matching matching algorithms~\cite{chiragdoi:10.1089/cmb.2022.0266} possibly for long read alignment to either a reference pangenome or a linear reference genome~\cite{10.1093/bioinformatics/bts414}. Lastly, genome to genome or genome to pangenome long LEMs detection may find similar sequences in the genomes on different genomic regions. MEMs detection may miss these similar sequences on different genomic regions because these matches will typically be overshadowed by larger encompassing matches that occur in roughly the same region in the pattern and the text. The long LEMs may therefore reveal old structural variants that MEMs and general alignment algorithms are both unable to reveal. MEMs don't reveal these variants due to looking for only the largest matches in a region on the pattern while alignment algorithms don't due to better alignments existing in closeby genomic regions or alignment algorithms being too computationally expensive to run on very large datasets.

Overall, we have provided a linear time algorithm for outputting all long LEMs of a pattern with respect to a text in BWT runs compressed space given the matching statistics of the pattern with respect to the text. We have also applied the move data structure of Nishimoto and Tabei to computation of $PLCP$ in constant time. Therefore, we can compute $LCP[i]$ given $SA[i]$ in constant time. We apply these results to modify the OptBWTR, creating OptBWTRL. OptBWTRL is an $O(r)$ space data structure that computes $\phi$ and $PLCP$ in constant time and long LEMs in linear time given matching statistics. These algorithms also result in a linear time long LEM query algorithm in uncompressed string indexes.



\bibliography{main}

\begin{thebibliography}{10}

\bibitem{ahmed2023spumoni}
Omar~Y Ahmed, Massimiliano Rossi, Travis Gagie, Christina Boucher, and Ben
  Langmead.
\newblock Spumoni 2: improved classification using a pangenome index of
  minimizer digests.
\newblock {\em Genome Biology}, 24(1):122, 2023.

\bibitem{thorup_10.1145/1236457.1236460}
Arne Andersson and Mikkel Thorup.
\newblock Dynamic ordered sets with exponential search trees.
\newblock {\em J. ACM}, 54(3):13–es, June 2007.
\newblock \href {https://doi.org/10.1145/1236457.1236460}
  {\path{doi:10.1145/1236457.1236460}}.

\bibitem{10.1007/978-3-319-19929-0_3}
Djamal Belazzougui, Fabio Cunial, Travis Gagie, Nicola Prezza, and Mathieu
  Raffinot.
\newblock Composite repetition-aware data structures.
\newblock In Ferdinando Cicalese, Ely Porat, and Ugo Vaccaro, editors, {\em
  Combinatorial Pattern Matching}, pages 26--39, Cham, 2015. Springer
  International Publishing.

\bibitem{blocktreeBELAZZOUGUI20211}
Djamal Belazzougui, Manuel Cáceres, Travis Gagie, Paweł Gawrychowski, Juha
  Kärkkäinen, Gonzalo Navarro, Alberto Ordóñez, Simon~J. Puglisi, and Yasuo
  Tabei.
\newblock Block trees.
\newblock {\em Journal of Computer and System Sciences}, 117:1--22, 2021.
\newblock \href {https://doi.org/https://doi.org/10.1016/j.jcss.2020.11.002}
  {\path{doi:https://doi.org/10.1016/j.jcss.2020.11.002}}.

\bibitem{bender2022optimal}
Michael~A. Bender, Mart\'{\i}n Farach-Colton, John Kuszmaul, William Kuszmaul,
  and Mingmou Liu.
\newblock On the optimal time/space tradeoff for hash tables.
\newblock In {\em Proceedings of the 54th Annual ACM SIGACT Symposium on Theory
  of Computing}, STOC 2022, page 1284–1297, New York, NY, USA, 2022.
  Association for Computing Machinery.
\newblock \href {https://doi.org/10.1145/3519935.3519969}
  {\path{doi:10.1145/3519935.3519969}}.

\bibitem{bendersimpledoi:10.1137/1.9781611977936.33}
Michael~A. Bender, Martín Farach-Colton, John Kuszmaul, and William Kuszmaul.
\newblock {\em Modern Hashing Made Simple}, pages 363--373.
\newblock Society for Industrial and Applied Mathematics, 2024.
\newblock \href {https://doi.org/10.1137/1.9781611977936.33}
  {\path{doi:10.1137/1.9781611977936.33}}.

\bibitem{all2024genomic}
Alexander~G. Bick, Ginger~A. Metcalf, Kelsey~R. Mayo, Lee Lichtenstein, Shimon
  Rura, Robert~J. Carroll, Anjene Musick, Jodell~E. Linder, I.~King Jordan,
  Shashwat~Deepali Nagar, Shivam Sharma, Robert Meller, Melissa Basford, Eric
  Boerwinkle, Mine~S. Cicek, Kimberly~F. Doheny, Evan~E. Eichler, Stacey
  Gabriel, Richard~A. Gibbs, David Glazer, Paul~A. Harris, Gail~P. Jarvik,
  Anthony Philippakis, Heidi~L. Rehm, Dan~M. Roden, Stephen~N. Thibodeau, Scott
  Topper, Ashley~L. Blegen, Samantha~J. Wirkus, Victoria~A. Wagner, Jeffrey~G.
  Meyer, Donna~M. Muzny, Eric Venner, Michelle~Z. Mawhinney, Sean M.~L.
  Griffith, Elvin Hsu, Hua Ling, Marcia~K. Adams, Kimberly Walker, Jianhong Hu,
  Harsha Doddapaneni, Christie~L. Kovar, Mullai Murugan, Shannon Dugan, Ziad
  Khan, Niall~J. Lennon, Christina Austin-Tse, Eric Banks, Michael Gatzen,
  Namrata Gupta, Emma Henricks, Katie Larsson, Sheli McDonough, Steven~M.
  Harrison, Christopher Kachulis, Matthew~S. Lebo, Cynthia~L. Neben, Marcie
  Steeves, Alicia~Y. Zhou, Joshua~D. Smith, Christian~D. Frazar, Colleen~P.
  Davis, Karynne~E. Patterson, Marsha~M. Wheeler, Sean McGee, Christina~M.
  Lockwood, Brian~H. Shirts, Colin~C. Pritchard, Mitzi~L. Murray, Valeria
  Vasta, Dru Leistritz, Matthew~A. Richardson, Jillian~G. Buchan, Aparna
  Radhakrishnan, Niklas Krumm, Brenna~W. Ehmen, Sophie Schwartz, M.~Morgan~T.
  Aster, Kristian Cibulskis, Andrea Haessly, Rebecca Asch, Aurora Cremer, Kylee
  Degatano, Akum Shergill, Laura~D. Gauthier, Samuel~K. Lee, Aaron Hatcher,
  George~B. Grant, Genevieve~R. Brandt, Miguel Covarrubias, Ashley Able,
  Ashley~E. Green, Jennifer Zhang, Henry~R. Condon, Yuanyuan Wang, Moira~K.
  Dillon, C.~H. Albach, Wail Baalawi, Seung~Hoan Choi, Xin Wang, Elisabeth~A.
  Rosenthal, Andrea~H. Ramirez, Sokny Lim, Siddhartha Nambiar, Bradley
  Ozenberger, Anastasia~L. Wise, Chris Lunt, Geoffrey~S. Ginsburg, Joshua~C.
  Denny, The~All of~Us~Research Program Genomics~Investigators,
  Manuscript~Writing Group, All of~Us~Research Program Genomics
  Principal~Investigators, Mayo Biobank, Genome Center: Baylor-Hopkins
  Clinical~Genome Center,  Color Genome Center:~Broad, Mass
  General~Brigham Laboratory for Molecular~Medicine, Genome Center:~University
  of~Washington, {Data}, Research Center, All of~Us~Research Demonstration
  Project~Teams, and NIH~All of~Us~Research Program~Staff.
\newblock Genomic data in the all of us research program.
\newblock {\em Nature}, 627(8003):340--346, Mar 2024.
\newblock \href {https://doi.org/10.1038/s41586-023-06957-x}
  {\path{doi:10.1038/s41586-023-06957-x}}.

\bibitem{smempbwt10.1007/978-3-031-43980-3_8}
Paola Bonizzoni, Christina Boucher, Davide Cozzi, Travis Gagie, Dominik
  K{\"o}ppl, and Massimiliano Rossi.
\newblock Data structures for smem-finding in the pbwt.
\newblock In Franco~Maria Nardini, Nadia Pisanti, and Rossano Venturini,
  editors, {\em String Processing and Information Retrieval}, pages 89--101,
  Cham, 2023. Springer Nature Switzerland.

\bibitem{brown_et_al:LIPIcs.SEA.2022.16}
Nathaniel~K. Brown, Travis Gagie, and Massimiliano Rossi.
\newblock {RLBWT Tricks}.
\newblock In Christian Schulz and Bora U\c{c}ar, editors, {\em 20th
  International Symposium on Experimental Algorithms (SEA 2022)}, volume 233 of
  {\em Leibniz International Proceedings in Informatics (LIPIcs)}, pages
  16:1--16:16, Dagstuhl, Germany, 2022. Schloss Dagstuhl -- Leibniz-Zentrum
  f{\"u}r Informatik.
\newblock \href {https://doi.org/10.4230/LIPIcs.SEA.2022.16}
  {\path{doi:10.4230/LIPIcs.SEA.2022.16}}.

\bibitem{nate10.1007/978-3-031-90252-9_12}
Nathaniel~K. Brown, Vikram~S. Shivakumar, and Ben Langmead.
\newblock Improved pangenomic classification accuracy with chain statistics.
\newblock In Sriram Sankararaman, editor, {\em Research in Computational
  Molecular Biology}, pages 190--208, Cham, 2025. Springer Nature Switzerland.

\bibitem{burrows1994block}
Michael Burrows.
\newblock A block-sorting lossless data compression algorithm.
\newblock {\em SRS Research Report}, 124, 1994.

\bibitem{draftpangenomerelease2}
Human Pangenome~Reference Consortium.
\newblock Hprc data release 2 [online].
\newblock URL: \url{https://humanpangenome.org/hprc-data-release-2/}.

\bibitem{mupbwt10.1093/bioinformatics/btad552}
Davide Cozzi, Massimiliano Rossi, Simone Rubinacci, Travis Gagie, Dominik
  Köppl, Christina Boucher, and Paola Bonizzoni.
\newblock $\mu$-pbwt: a lightweight r-indexing of the pbwt for storing and
  querying uk biobank data.
\newblock {\em Bioinformatics}, 39(9):btad552, 09 2023.
\newblock \href {https://doi.org/10.1093/bioinformatics/btad552}
  {\path{doi:10.1093/bioinformatics/btad552}}.

\bibitem{demaine10.1007/11682462_34}
Erik~D. Demaine, Friedhelm Meyer~auf der Heide, Rasmus Pagh, and Mihai
  P{\v{a}}tra{\c{s}}cu.
\newblock De dictionariis dynamicis pauco spatio utentibus.
\newblock In Jos{\'e}~R. Correa, Alejandro Hevia, and Marcos Kiwi, editors,
  {\em LATIN 2006: Theoretical Informatics}, pages 349--361, Berlin,
  Heidelberg, 2006. Springer Berlin Heidelberg.

\bibitem{Depuydt2025.02.25.640119}
Lore Depuydt, Omar~Y. Ahmed, Jan Fostier, Ben Langmead, and Travis Gagie.
\newblock Run-length compressed metagenomic read classification with
  smem-finding and tagging.
\newblock {\em bioRxiv}, 2025.
\newblock \href {https://doi.org/10.1101/2025.02.25.640119}
  {\path{doi:10.1101/2025.02.25.640119}}.

\bibitem{bmovedepuydt_et_al:LIPIcs.WABI.2024.10}
Lore Depuydt, Luca Renders, Simon Van~de Vyver, Lennart Veys, Travis Gagie, and
  Jan Fostier.
\newblock {b-move: Faster Bidirectional Character Extensions in a Run-Length
  Compressed Index}.
\newblock In Solon~P. Pissis and Wing-Kin Sung, editors, {\em 24th
  International Workshop on Algorithms in Bioinformatics (WABI 2024)}, volume
  312 of {\em Leibniz International Proceedings in Informatics (LIPIcs)}, pages
  10:1--10:18, Dagstuhl, Germany, 2024. Schloss Dagstuhl -- Leibniz-Zentrum
  f{\"u}r Informatik.
\newblock \href {https://doi.org/10.4230/LIPIcs.WABI.2024.10}
  {\path{doi:10.4230/LIPIcs.WABI.2024.10}}.

\bibitem{PBWT}
Richard Durbin.
\newblock Efficient haplotype matching and storage using the positional
  burrows–wheeler transform (pbwt).
\newblock {\em Bioinformatics}, 30(9):1266--1272, 01 2014.
\newblock \href {https://doi.org/10.1093/bioinformatics/btu014}
  {\path{doi:10.1093/bioinformatics/btu014}}.

\bibitem{FMIndex10.1145/1082036.1082039}
Paolo Ferragina and Giovanni Manzini.
\newblock Indexing compressed text.
\newblock {\em J. ACM}, 52(4):552–581, July 2005.
\newblock \href {https://doi.org/10.1145/1082036.1082039}
  {\path{doi:10.1145/1082036.1082039}}.

\bibitem{rindex10.1145/3375890}
Travis Gagie, Gonzalo Navarro, and Nicola Prezza.
\newblock Fully functional suffix trees and optimal text searching in bwt-runs
  bounded space.
\newblock {\em J. ACM}, 67(1), January 2020.
\newblock \href {https://doi.org/10.1145/3375890} {\path{doi:10.1145/3375890}}.

\bibitem{chiragdoi:10.1089/cmb.2022.0266}
Chirag Jain, Daniel Gibney, and Sharma~V. Thankachan.
\newblock Algorithms for colinear chaining with overlaps and gap costs.
\newblock {\em Journal of Computational Biology}, 29(11):1237--1251, 2022.
\newblock \href {https://doi.org/10.1089/cmb.2022.0266}
  {\path{doi:10.1089/cmb.2022.0266}}.

\bibitem{10.1007/978-3-642-02441-2_17}
Juha K{\"a}rkk{\"a}inen, Giovanni Manzini, and Simon~J. Puglisi.
\newblock Permuted longest-common-prefix array.
\newblock In Gregory Kucherov and Esko Ukkonen, editors, {\em Combinatorial
  Pattern Matching}, pages 181--192, Berlin, Heidelberg, 2009. Springer Berlin
  Heidelberg.

\bibitem{KO2005143}
Pang Ko and Srinivas Aluru.
\newblock Space efficient linear time construction of suffix arrays.
\newblock {\em Journal of Discrete Algorithms}, 3(2):143--156, 2005.
\newblock Combinatorial Pattern Matching (CPM) Special Issue.
\newblock \href {https://doi.org/https://doi.org/10.1016/j.jda.2004.08.002}
  {\path{doi:https://doi.org/10.1016/j.jda.2004.08.002}}.

\bibitem{delta10.1007/978-3-030-61792-9_17}
Tomasz Kociumaka, Gonzalo Navarro, and Nicola Prezza.
\newblock Towards a definitive measure of repetitiveness.
\newblock In Yoshiharu Kohayakawa and Fl{\'a}vio~Keidi Miyazawa, editors, {\em
  LATIN 2020: Theoretical Informatics}, pages 207--219, Cham, 2020. Springer
  International Publishing.

\bibitem{langmead2012fast}
Ben Langmead and Steven~L Salzberg.
\newblock Fast gapped-read alignment with bowtie 2.
\newblock {\em Nature methods}, 9(4):357--359, 2012.

\bibitem{liforwardbackward10.1093/bioinformatics/bts280}
Heng Li.
\newblock Exploring single-sample snp and indel calling with whole-genome de
  novo assembly.
\newblock {\em Bioinformatics}, 28(14):1838--1844, 05 2012.
\newblock \href {https://doi.org/10.1093/bioinformatics/bts280}
  {\path{doi:10.1093/bioinformatics/bts280}}.

\bibitem{ropebwt310.1093/bioinformatics/btae717}
Heng Li.
\newblock Bwt construction and search at the terabase scale.
\newblock {\em Bioinformatics}, 40(12):btae717, 11 2024.
\newblock \href {https://doi.org/10.1093/bioinformatics/btae717}
  {\path{doi:10.1093/bioinformatics/btae717}}.

\bibitem{bwa10.1093/bioinformatics/btp324}
Heng Li and Richard Durbin.
\newblock Fast and accurate short read alignment with burrows–wheeler
  transform.
\newblock {\em Bioinformatics}, 25(14):1754--1760, 05 2009.
\newblock \href {https://doi.org/10.1093/bioinformatics/btp324}
  {\path{doi:10.1093/bioinformatics/btp324}}.

\bibitem{10.1093/bioinformatics/btp698}
Heng Li and Richard Durbin.
\newblock Fast and accurate long-read alignment with burrows–wheeler
  transform.
\newblock {\em Bioinformatics}, 26(5):589--595, 01 2010.
\newblock \href {https://doi.org/10.1093/bioinformatics/btp698}
  {\path{doi:10.1093/bioinformatics/btp698}}.

\bibitem{UKBLi2023.12.06.23299426}
Shuwei Li, Keren~J Carss, Bjarni~V Halldorsson, Adrian Cortes, and UK~Biobank
  Whole-Genome~Sequencing Consortium.
\newblock Whole-genome sequencing of half-a-million uk biobank participants.
\newblock {\em medRxiv}, 2023.
\newblock \href {https://doi.org/10.1101/2023.12.06.23299426}
  {\path{doi:10.1101/2023.12.06.23299426}}.

\bibitem{liao2023draft}
Wen-Wei Liao, Mobin Asri, Jana Ebler, Daniel Doerr, Marina Haukness, Glenn
  Hickey, Shuangjia Lu, Julian~K Lucas, Jean Monlong, Haley~J Abel, et~al.
\newblock A draft human pangenome reference.
\newblock {\em Nature}, 617(7960):312--324, 2023.

\bibitem{10.1093/bioinformatics/bts414}
Yongchao Liu and Bertil Schmidt.
\newblock Long read alignment based on maximal exact match seeds.
\newblock {\em Bioinformatics}, 28(18):i318--i324, 09 2012.
\newblock \href {https://doi.org/10.1093/bioinformatics/bts414}
  {\path{doi:10.1093/bioinformatics/bts414}}.

\bibitem{doi:10.1089/cmb.2009.0169}
Veli M\"{a}kinen, Gonzalo Navarro, Jouni Sir\'{e}n, and Niko V\"{a}lim\"{a}ki.
\newblock Storage and retrieval of highly repetitive sequence collections.
\newblock {\em Journal of Computational Biology}, 17(3):281--308, 2010.
\newblock \href {https://doi.org/10.1089/cmb.2009.0169}
  {\path{doi:10.1089/cmb.2009.0169}}.

\bibitem{miga2021need}
Karen~H Miga and Ting Wang.
\newblock The need for a human pangenome reference sequence.
\newblock {\em Annual Review of Genomics and Human Genetics}, 22(1):81--102,
  2021.

\bibitem{naseri10.1093/bioinformatics/btz347}
Ardalan Naseri, Erwin Holzhauser, Degui Zhi, and Shaojie Zhang.
\newblock Efficient haplotype matching between a query and a panel for
  genealogical search.
\newblock {\em Bioinformatics}, 35(14):i233--i241, 07 2019.
\newblock \href {https://doi.org/10.1093/bioinformatics/btz347}
  {\path{doi:10.1093/bioinformatics/btz347}}.

\bibitem{navarropart1_10.1145/3434399}
Gonzalo Navarro.
\newblock Indexing highly repetitive string collections, part i: Repetitiveness
  measures.
\newblock {\em ACM Comput. Surv.}, 54(2), March 2021.
\newblock \href {https://doi.org/10.1145/3434399} {\path{doi:10.1145/3434399}}.

\bibitem{nishimoto_et_al:LIPIcs.ICALP.2021.101}
Takaaki Nishimoto and Yasuo Tabei.
\newblock {Optimal-Time Queries on BWT-Runs Compressed Indexes}.
\newblock In Nikhil Bansal, Emanuela Merelli, and James Worrell, editors, {\em
  48th International Colloquium on Automata, Languages, and Programming (ICALP
  2021)}, volume 198 of {\em Leibniz International Proceedings in Informatics
  (LIPIcs)}, pages 101:1--101:15, Dagstuhl, Germany, 2021. Schloss Dagstuhl --
  Leibniz-Zentrum f{\"u}r Informatik.
\newblock \href {https://doi.org/10.4230/LIPIcs.ICALP.2021.101}
  {\path{doi:10.4230/LIPIcs.ICALP.2021.101}}.

\bibitem{rrr10.1007/3-540-45061-0_30}
Rajeev Raman and Satti~Srinivasa Rao.
\newblock Succinct dynamic dictionaries and trees.
\newblock In Jos C.~M. Baeten, Jan~Karel Lenstra, Joachim Parrow, and
  Gerhard~J. Woeginger, editors, {\em Automata, Languages and Programming},
  pages 357--368, Berlin, Heidelberg, 2003. Springer Berlin Heidelberg.

\bibitem{monidoi:10.1089/cmb.2021.0290}
Massimiliano Rossi, Marco Oliva, Ben Langmead, Travis Gagie, and Christina
  Boucher.
\newblock Moni: A pangenomic index for finding maximal exact matches.
\newblock {\em Journal of Computational Biology}, 29(2):169--187, 2022.
\newblock PMID: 35041495.
\newblock \href {https://doi.org/10.1089/cmb.2021.0290}
  {\path{doi:10.1089/cmb.2021.0290}}.

\bibitem{gbwtquerySanaullah2025.02.03.634410}
Ahsan Sanaullah, Seba Villalobos, Degui Zhi, and Shaojie Zhang.
\newblock Haplotype matching with gbwt for pangenome graphs.
\newblock {\em bioRxiv}, 2025.
\newblock \href {https://doi.org/10.1101/2025.02.03.634410}
  {\path{doi:10.1101/2025.02.03.634410}}.

\bibitem{dpbwt}
Ahsan Sanaullah, Degui Zhi, and Shaojie Zhang.
\newblock d-pbwt: dynamic positional burrows–wheeler transform.
\newblock {\em Bioinformatics}, 37(16):2390--2397, 02 2021.
\newblock \href {https://doi.org/10.1093/bioinformatics/btab117}
  {\path{doi:10.1093/bioinformatics/btab117}}.

\bibitem{dynmupbwt10.1007/978-3-031-90252-9_13}
Pramesh Shakya, Ahsan Sanaullah, Degui Zhi, and Shaojie Zhang.
\newblock Dynamic $\mu$-pbwt: Dynamic run-length compressed pbwt for biobank
  scale data.
\newblock In Sriram Sankararaman, editor, {\em Research in Computational
  Molecular Biology}, pages 209--226, Cham, 2025. Springer Nature Switzerland.

\bibitem{singh2022reference}
Vipin Singh, Shweta Pandey, and Anshu Bhardwaj.
\newblock From the reference human genome to human pangenome: Premise, promise
  and challenge.
\newblock {\em Frontiers in Genetics}, 13:1042550, 2022.

\bibitem{GBWT}
Jouni Sirén, Erik Garrison, Adam~M Novak, Benedict Paten, and Richard Durbin.
\newblock Haplotype-aware graph indexes.
\newblock {\em Bioinformatics}, 36(2):400--407, 07 2019.
\newblock \href {https://doi.org/10.1093/bioinformatics/btz575}
  {\path{doi:10.1093/bioinformatics/btz575}}.

\bibitem{song2024centrifuger}
Li~Song and Ben Langmead.
\newblock Centrifuger: lossless compression of microbial genomes for efficient
  and accurate metagenomic sequence classification.
\newblock {\em Genome Biology}, 25(1):106, 2024.

\bibitem{taylor2024beyond}
Dylan~J Taylor, Jordan~M Eizenga, Qiuhui Li, Arun Das, Katharine~M Jenike,
  Eimear~E Kenny, Karen~H Miga, Jean Monlong, Rajiv~C McCoy, Benedict Paten,
  et~al.
\newblock Beyond the human genome project: the age of complete human genome
  sequences and pangenome references.
\newblock {\em Annual review of genomics and human genetics}, 25, 2024.

\bibitem{obit_10.1145/2447712.2447737}
Mikkel Thorup.
\newblock Mihai pǎtra\c{s}cu: obituary and open problems.
\newblock {\em SIGACT News}, 44(1):110–114, March 2013.
\newblock \href {https://doi.org/10.1145/2447712.2447737}
  {\path{doi:10.1145/2447712.2447737}}.

\bibitem{sylpbwt10.1093/bioinformatics/btac734}
Victor Wang, Ardalan Naseri, Shaojie Zhang, and Degui Zhi.
\newblock Syllable-pbwt for space-efficient haplotype long-match query.
\newblock {\em Bioinformatics}, 39(1):btac734, 11 2022.
\newblock \href {https://doi.org/10.1093/bioinformatics/btac734}
  {\path{doi:10.1093/bioinformatics/btac734}}.

\bibitem{zakeri2024movi}
Mohsen Zakeri, Nathaniel~K Brown, Omar~Y Ahmed, Travis Gagie, and Ben Langmead.
\newblock Movi: a fast and cache-efficient full-text pangenome index.
\newblock {\em iScience}, 27(12), 2024.

\end{thebibliography}

\end{document}